\documentclass[acmsmall]{acmart}
\acmJournal{PACMPL}
\acmYear{2021}
\acmISBN{} 
\acmDOI{} 
\startPage{1}
\setcopyright{none}

\usepackage{tikz}
\usetikzlibrary{arrows,automata,positioning,backgrounds,calc,patterns}

\usepackage{mathpartir}

\newcommand{\figref}[1]{Figure~\ref{fig:#1}}
\newcommand{\sectref}[1]{\S\ref{sec:#1}}

\newcommand{\ie}{i.e., }

\newcommand{\key}{k}
\newcommand{\Keys}{\mathsf{Keys}}

\newcommand{\CO}{\textcolor{red}{\textsf{co} }}

\newcommand{\set}[1]{{\{{#1}\}}}

\newcommand{\tup}[1]{{\left\langle{#1}\right\rangle}}
\renewcommand{\implies}{\Rightarrow}

\newcommand{\Vars}{\mathsf{Vars}}
\newcommand{\Val}{\mathsf{Vals}}
\newcommand{\Vals}{\mathsf{Vals}}
\newcommand{\Columns}{\mathbb{C}}
\newcommand{\Rows}{\mathbb{R}}
\newcommand{\Tables}{\mathbb{T}}
\newcommand{\OId}{\mathsf{OpId}}
\newcommand{\Op}{\mathsf{Op}}
\newcommand{\xvar}{{x}}
\newcommand{\prog}{{\mathsf{P}}}
\newcommand{\histOf}[2][]{{\mathsf{hist}_{#1}({#2})}}

\newcommand{\yvar}{{y}}
\newcommand{\val}{{v}}
\newcommand{\id}{{i}}
\newcommand{\rd}[3][]{\textsf{read}_{#1}({#2},{#3})}
\newcommand{\rdo}{\textsf{read}}
\newcommand{\wrt}[3][]{\textsf{write}_{#1}({#2},{#3})}

\newcommand{\tr}{{t}}

\newcommand{\hist}{{h}}
\newcommand{\po}{\textcolor{red}{\mathsf{po}}}
\newcommand{\so}{\textcolor{red}{\mathsf{so}}}
\newcommand{\co}{\textcolor{red}{\mathsf{co}}}
\newcommand{\wro}[1][]{\textcolor{red}{\mathsf{wr}_{#1}}}

\newcommand{\readOp}[1]{\mathsf{reads}({#1})}
\newcommand{\tlogs}[1]{\mathsf{TLogs}({#1})}
\newcommand{\transC}[1]{\mathsf{compTrans}({#1})}

\newcommand{\writeOp}[1]{\mathsf{writes}({#1})}
\newcommand{\writeVar}[2]{{#1}\ \mathsf{writes}\ {#2}}

\renewcommand{\hist}{{h}}

\newcommand{\btrue}{\mathsf{true}}
\newcommand{\bfalse}{\mathsf{false}}

\newcommand{\eqdef}{::=}
\newcommand{\tab}{\mathit{tab}}
\newcommand{\pkey}{\mathit{pkey}}
\newcommand{\pkeyVal}{\mathit{pkeyVal}}
\newcommand{\ibegin}{\mathtt{begin}}
\newcommand{\iadd}{\mathtt{add}}
\newcommand{\iremove}{\mathtt{remove}}
\newcommand{\ielements}{\mathtt{elements}}
\newcommand{\icontains}{\mathtt{has}}
\newcommand{\icommit}{\mathtt{commit}}
\newcommand{\iif}[2]{\mathtt{if}({#1}) \{ {#2} \}}
\newcommand{\iread}{\mathtt{read}}
\newcommand{\iwrite}{\mathtt{write}}
\newcommand{\iselect}[4]{\mathtt{SELECT}\ {#1}\ \mathtt{AS}\ {#2}\ \mathtt{FROM}\ {#3}\ \mathtt{WHERE}\ {#4}}
\newcommand{\iinsert}[2]{\mathtt{INSERT\ INTO}\ {#1}\ \mathtt{VALUES}\ {#2}}
\newcommand{\idelete}[2]{\mathtt{DELETE\ FROM}\ {#1}\ \mathtt{WHERE}\ {#2}}
\newcommand{\iupdate}[3]{\mathtt{UPDATE}\ {#1}\ \mathtt{SET}\ {#2}\ \mathtt{WHERE}\ {#3}}
\newcommand{\KVProgs}{\mathcal{P}_{KV}}
\newcommand{\SQLProgs}{\mathcal{P}_{SQL}}
\newcommand{\DBschema}{\mathcal{S}}
\newcommand{\DBinst}{\mathcal{D}}

\usepackage{listings}
\lstdefinelanguage{MyLang}{%
  keywords = { delete, do, each, else, export, finally, for, foreach, function, if, in, let, of, return, void, while, with, yield, elements, read, write, insert, remove, add, AddItem, DeleteItem, Push, Pop, Enroll, Tweet, Timeline, NewsFeed, Begin, Commit},
  morecomment = [l]{//},
  morecomment = [s]{/*}{*/},
  morestring  = [b]',
  morestring  = [b]",
  sensitive   = true,
}

\lstdefinelanguage{Java10}{
  language      = Java,
  morekeywords  ={ var },
}

\usepackage{booktabs}   
\usepackage{subcaption} 

\def\tool{\mbox{MonkeyDB}}

\begin{document}

\title[~]{MonkeyDB: Effectively Testing Correctness against Weak Isolation Levels}         
\author{Ranadeep Biswas}
\affiliation{
  \institution{IRIF, University of Paris \& CNRS} 
  \country{France}                    
}
\email{ranadeep@irif.fr}          

\author{Diptanshu Kakwani}
\affiliation{
  \institution{IIT Madras} 
  \country{India}                    
}
\email{dipk@cse.iitm.ac.in}          

\author{Jyothi Vedurada}
\affiliation{
  \institution{Microsoft Research} 
  \country{India}                    
}
\email{t-vevedu@microsoft.com}          

\author{Constantin Enea}
\affiliation{
  \institution{IRIF, University of Paris \& CNRS} 
  \country{France}                    
}
\email{cenea@irif.fr}          

\author{Akash Lal}
\affiliation{
  \institution{Microsoft Research} 
  \country{India}                    
}
\email{akashl@microsoft.com}          

\begin{abstract}
Modern applications, such as social networking systems and e-commerce platforms
are centered around using large-scale storage systems for storing and retrieving
data. In the presence of concurrent accesses, these storage systems trade off isolation
for performance. The weaker the isolation level, the more behaviors a storage
system is allowed to exhibit and it is up to the developer to ensure that their
application can tolerate those behaviors. However, these weak behaviors only
occur rarely in practice, that too outside the control of the application, 
making it difficult for developers to test the robustness of their 
code against weak isolation levels. 

This paper presents MonkeyDB, a mock storage system for testing storage-backed
applications. MonkeyDB supports a Key-Value interface as well as SQL queries
under multiple isolation levels. It uses a logical specification of the isolation
level to compute, on a read operation, the set of all possible return values.
MonkeyDB then returns a value randomly from this set. 
We show that MonkeyDB provides 
good coverage of weak behaviors, which is complete in the limit. We test a
variety of applications for assertions that fail only under weak isolation.
MonkeyDB is able to break each of those assertions in a small number of attempts. 

\end{abstract}


\maketitle

\sloppy
\section{Introduction}
\label{sec:intro}

Data storage is no longer about writing data to a single
disk with a single point of access. Modern applications require not just data
reliability, but also high-throughput concurrent accesses. 
Applications concerning supply chains, banking, etc. use traditional relational databases
for storing and processing data, whereas applications such as social networking
software and e-commerce platforms 
use cloud-based storage systems (such as Azure CosmosDb \cite{cosmosdb}, Amazon DynamoDb
\cite{decandia2007dynamo}, Facebook TAO \cite{facebook-tao}, etc.). We use the term \textit{storage
system} in this paper to refer to any such database system/service.

 

Providing high-throughput processing, unfortunately, comes at an unavoidable cost of weakening 
the guarantees offered to users.
Concurrently-connected clients may end up observing different views of the same data. 
These ``anomalies'' can be prevented by using a strong \textit{isolation level} 
such as \textit{serializability}, which essentially offers a single view of the
data. However, serializability requires expensive synchronization and incurs a high performance cost. As a
consequence, most storage systems use weaker isolation levels, such as 
{\it Causal Consistency}~\cite{DBLP:journals/cacm/Lamport78,DBLP:conf/sosp/LloydFKA11,antidote},
{\it Snapshot Isolation}~\cite{DBLP:conf/sigmod/BerensonBGMOO95}, {\it Read
Committed}~\cite{DBLP:conf/sigmod/BerensonBGMOO95}, etc. for better performance.
In a recent survey of
database administrators \cite{survey}, 86\% of the participants responded that
most or all of the transactions in their databases execute at read committed isolation level.

\begin{figure}
	\begin{minipage}{4.2cm}
		\begin{lstlisting}[basicstyle=\ttfamily\footnotesize,escapeinside={(*}{*)},language=MyLang]
// Append item to cart
AddItem(item i, userId) {
  Begin()
  key = "cart:" + userId
  cart = read(key)
  cart.append(i)
  write(key, cart)
  Commit()
}
		\end{lstlisting}
	\end{minipage}
	\hspace{-5mm}
	\begin{minipage}{4.2cm}
		\begin{lstlisting}[xleftmargin=4mm,basicstyle=\ttfamily\footnotesize,escapeinside={(*}{*)},language=MyLang]
// Fetch cart and delete item
DeleteItem(item i, userId) {
  Begin()
  key = "cart:" + userId
  cart = read(key)
  cart.remove(i)
  write(key, cart)
  Commit()
}
		\end{lstlisting}
	\end{minipage}
	
\vspace{-6mm}	
  \resizebox{8.5cm}{!}{
   \begin{tikzpicture}[->,>=stealth',shorten >=1pt,auto,node distance=4cm,
     semithick, transform shape]
    \node (s11l) at (1.15, 2.1) {\textbf{Initial state}};
    \node[draw, rounded corners=2mm] (t0) at (2.05, 1.5) {\begin{tabular}{l} $\wrt{\texttt{cart:}u}{\{..\, I\, ..\}}$ \end{tabular}};
    \node[draw, rounded corners=2mm, minimum width=3.6cm, minimum height=1.3cm] (s1) at (0, -0.1) {};
    \node[style={inner sep=0,outer sep=0}] (s11) at (0, 0.3) {\begin{tabular}{l} $\rd{\texttt{cart:}u}{\{..\, I\, ..\}}$\end{tabular}};
    \node[style={inner sep=0,outer sep=0}] (s12) at (0, -0.5) {\begin{tabular}{l} $\wrt{\texttt{cart:}u}{\{..\, I,I\, ..\}}$ \end{tabular}};
    \node (s11l) at (-1, 0.8) {\textbf{AddItem}};
    \node[draw, rounded corners=2mm, minimum width=3.6cm, minimum height=1.3cm] (s2) at (4.1, -0.1) {};
    \node[style={inner sep=0,outer sep=0}] (s21) at (4.1, 0.3) {\begin{tabular}{l} $\rd{\texttt{cart:}u}{\{..\, I\, ..\}}$ \end{tabular}};
    \node[style={inner sep=0,outer sep=0}] (s22) at (4.1, -0.5) {\begin{tabular}{l} $\wrt{\texttt{cart:}u}{\{..\, ..\}}$ \end{tabular}};
    \node (s11l) at (4.9, 0.8) {\textbf{DeleteItem}};
    \node[draw, rounded corners=2mm] (r1) at (8.3, 0) {\begin{tabular}{l} $\rd{\texttt{cart:}u}{\{..\, ..\}}$ \end{tabular}};
    \node[draw, rounded corners=2mm] (r2) at (8.3, -1.3) {\begin{tabular}{l} $\rd{\texttt{cart:}u}{\{..\, I, I\, .\}}$ \end{tabular}};
    \path (s11) edge[left] node {$\po$} (s12);
    \path (s21) edge[left] node {$\po$} (s22);
    \path (t0) edge[left] node {$\wro$} (s1);
    \path (t0) edge[right] node {$\wro$} (s2);
    \path (r1) edge[left] node {$\so$} (r2);
    \path (s2) edge[above] node {$\wro$} (r1);
    \path (s1) edge[below,bend right=11] node {$\wro$} (r2);
   \end{tikzpicture}  
  }

\vspace{-2mm}
	\caption{A simple shopping cart service.}
	\label{fig:motiv}
\vspace{-3mm}
\end{figure}
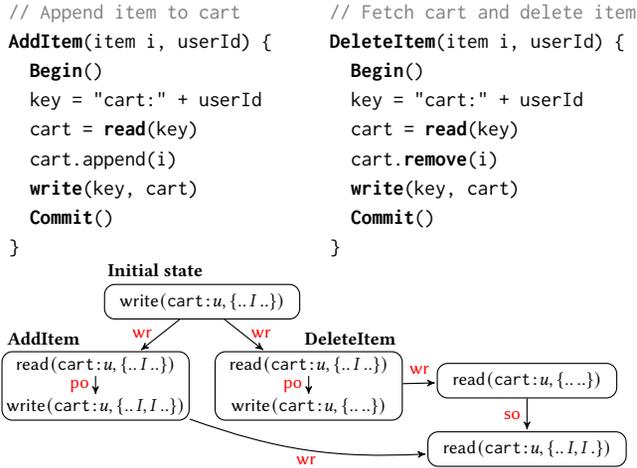

A weaker isolation level allows for more possible behaviors than stronger
isolation levels. It is up to the developers then to ensure that their
application can tolerate this larger set of behaviors. Unfortunately, weak
isolation levels are hard to understand or reason about
\cite{DBLP:conf/popl/BrutschyD0V17,adya-thesis} and resulting application bugs
can cause loss of business \cite{acidrain}.
Consider a simple shopping cart application that stores a per-client shopping
cart in a key-value store (\textit{key} is the client ID and \textit{value} is a
multi-set of items). \figref{motiv} shows procedures for adding an item to the cart
(\texttt{AddItem}) and deleting \textit{all} instances of an item from the cart
(\texttt{DeleteItem}). Each procedure executes in a transaction, represented by
the calls to \texttt{Begin} and \texttt{Commit}. Suppose that initially, a user $u$ has 
a single instance of item $I$ in their cart.
Then the user connects to the application via two different
sessions (for instance, via two browser windows), adds $I$ in one session
(\texttt{AddItem($I$, $u$)}) and deletes $I$ in the other session
(\texttt{DeleteItem($I$, $u$)}). With serializability, the cart can either be
left in the state $\{ I \}$ (delete happened first, followed by the add) or $\emptyset$ (delete
happened second). However, with causal consistency (or read committed), it is possible that with two
sequential reads of the shopping cart, the cart is empty in the first read
(signaling that the delete has succeeded), but there are \textit{two} instances of $I$ 
in the second read! Such anomalies, of deleted items reappearing, have been
noted in previous work \cite{decandia2007dynamo}. 

\paragraph{Testing storage-based applications}
This paper addresses the problem of \textit{testing} code for correctness
against weak behaviors: a developer should be able to write a test that runs
their application and then asserts for correct behavior. 
The main difficulty today is getting coverage of weak behaviors during
the test. If one runs the test
against the actual production storage system, it is very likely to only result in
serializable behaviors because of their optimized implementation. For
instance, only 0.0004\% of all reads performed on Facebook's TAO storage system 
were not serializable \cite{facebook-consistency}. 
Emulators, offered by cloud providers for local development, on the other hand, do not support weaker
isolation levels at all \cite{cosmosdb-local}. Another option, possible when
the storage system is available open-source, is to set it up with a 
tool like Jepsen~\cite{jepsen} to inject noise (bring down replicas or
delay packets on the network). 
This approach is unable to provide good coverage at the level of client operations
\cite{clotho} (\sectref{oltp}). Another line of work has focussed on finding
anomalies by identifying non-serializable behavior (\sectref{related}). Anomalies, however, do not
always correspond to bugs \cite{DBLP:conf/pldi/BrutschyD0V18,isodiff}; they may
either not be important (e.g., gather statistics) or may already be handled in
the application (e.g., checking and deleting duplicate items).

We present MonkeyDB, a mock in-memory storage system meant for testing
correctness of storage-backed applications. 
MonkeyDB supports 
common APIs for accessing data (key-value updates, as well as SQL queries),
making it an easy substitute for an actual storage system. MonkeyDB
can be configured with one of several isolation levels. 
On a read operation, MonkeyDB computes the set of all possible return values
allowed under the chosen isolation level, and randomly returns one of them. The
developer can then simply execute their test multiple times to get coverage of
possible weak behaviors. For the program in \figref{motiv}, if we write a test
asserting that two sequential reads cannot return empty-cart followed by $\{I,
I\}$, then it takes only 20 runs of the test (on average) to fail the
assert. In contrast, the test does not fail when using MySQL with read committed, 
even after 100k runs. 

\paragraph{Design of MonkeyDB}
MonkeyDB does not rely on stress generation, fault
injection, or data replication. 
Rather, it works directly with a formalization of
the given isolation level in order to compute allowed return values.

The theory behind MonkeyDB builds on the axiomatic definitions of isolation
levels introduced by Biswas et al. \cite{DBLP:journals/pacmpl/BiswasE19}. These
definitions use logical constraints (called \emph{axioms}) to characterize the
set of executions of a key-value store that conform to a particular isolation
level (we discuss SQL queries later). These constraints refer to a specific set of
relations between events/transactions in an execution that describe control-flow
or data-flow dependencies: a program order $\po$ between events in the same
transaction, a session order $\so$ between transactions in the same session\footnote{A
session is a sequential interface to the storage system. It corresponds to what
is also called a connection.}, and a write-read $\wro$ (read-from) relation that
associates each read event with a transaction that writes the value returned by
the read. These relations along with the events (also called, operations) in an
execution are called a \emph{history}. The history corresponding to the 
shopping cart anomaly explained above is given on the bottom of Figure~\ref{fig:motiv}.
Read operations include the read value, and boxes group events from the same transaction.
A history describes only the
interaction with the key-value store, omitting application side events (e.g., computing
the value to be written to a key). 

MonkeyDB implements a \emph{centralized} operational semantics for key-value stores, which is based on these axiomatic definitions. Transactions are executed \emph{serially}, one after another, the concurrency being simulated during the handling of read events.  
This semantics maintains a history that contains all the past events (from all
transactions/sessions), and write events are simply added to the history. The
value returned by a read event is established based on a non-deterministic
choice of a write-read dependency (concerning this read event) that satisfies
the axioms of the considered isolation level.
Depending on the weakness of the isolation
level, this makes it possible to return values written in arbitrarily ``old''
transactions, and simulate any concurrent behavior. For instance, the history in Figure~\ref{fig:motiv}
can be obtained by executing \texttt{AddItem}, \texttt{DeleteItem}, and then the two reads (serially).
The read in \texttt{DeleteItem} can take its value from the initial state and ``ignore'' the
previously executed \texttt{AddItem}, because the obtained history validates the axioms of 
causal consistency (or read committed). The same happens for the two later reads in the same
session, the first one being able to read from \texttt{DeleteItem} and the second one
from \texttt{AddItem}.

We formally prove that this semantics does indeed simulate any concurrent behavior, by 
showing that it is equivalent to a semantics where transactions are allowed to interleave.
In comparison with concrete implementations, this semantics makes it possible to handle 
a wide range of isolation levels in a uniform way. It only has two sources of
non-determinism: 
the order in which entire transactions are submitted, and the choice of write-read dependencies in read 
events. This enable better coverage of possible behaviors, the penalty in performance not
being an issue in safety testing workloads which are usually small (see our evaluation).

We also extend our semantics to cover SQL queries as well, by compiling SQL queries down to transactions with multiple key-value reads/writes. A table in a relational database is represented using a set of primary key values (identifying uniquely the set of rows) and a set of keys, one for each cell in the table. The set of primary key values is represented using a set of Boolean key-value pairs that simulate its characteristic function (adding or removing an element corresponds to updating one of these keys to $\btrue$ or $\bfalse$). Then, SQL queries are compiled to read or write accesses to the keys representing a table. For instance, a $\mathtt{SELECT}$ query that retrieves the set of rows in a table that satisfy a $\mathtt{WHERE}$ condition is compiled to (1) reading Boolean keys to identify the primary key values of the rows contained in the table, (2) reading keys that represent columns used in the $\mathtt{WHERE}$ condition, and (3) reading all the keys that represent cells in a row satisfying the $\mathtt{WHERE}$ condition. This rewriting contains the minimal set of accesses to the cells of a table that are needed to ensure the conventional specification of SQL.
It makes it possible to ``export'' formalizations of key-value store isolation levels to SQL transactions.


\paragraph{Contributions}

This paper makes the following contributions:
\begin{itemize}
\item We define an operational semantics for key-value stores under various
  isolation levels, which simulates all concurrent behaviors with executions
  where transactions execute serially (\sectref{op-kv}) and which is based 
  on the axiomatic definitions in~\cite{DBLP:journals/pacmpl/BiswasE19} (and outlined in \S\ref{sec:ax-kv}),
\item We broaden the scope of the key-value store semantics to SQL transactions
  using a compiler that rewrites SQL queries to key-value accesses (\sectref{SQL-to-KV}),
\item The operational semantics and the SQL compiler are implemented in a tool
  called MonkeyDB (\sectref{impl}). It randomly resolves possible choices to provide coverage
  of weak behaviors. It supports both a key-value interface as well as SQL,
  making it readily compatible with any storage-backed application.
\item We present an evaluation of MonkeyDB on several applications, showcasing
its superior coverage of weak behaviors as well as bug-finding abilities
(\sectref{micro}, \sectref{oltp}).\footnote{Source code of our benchmarks is available
as supplementary material.}
\end{itemize}

\section{Programming Language}

\begin{figure}
\small
\begin{align*}
\key\in \Keys\quad \xvar\in\Vars\quad \tab\in\Tables\quad \vec{c},\vec{c_1},\vec{c_2}\in \Columns^*
\end{align*}
\begin{align*}
\mathsf{Prog} &  \eqdef  \mathsf{Sess} \ \mid\  \mathsf{Sess}\,||\,\mathsf{Prog} \\
\mathsf{Sess} & \eqdef  \mathsf{Trans} \ \mid\  \mathsf{Trans}; \mathsf{Sess} \\
\mathsf{Trans} & \eqdef  \ibegin; \mathsf{Body}; \icommit \\
\mathsf{Body} & \eqdef  \mathsf{Instr} \ \mid\  \mathsf{Instr}; \mathsf{Body} \\
\mathsf{Instr} & \eqdef  \mathsf{InstrKV} \ \mid\  \mathsf{InstrSQL}\ \mid\  x := e \mid\ \iif{\phi(\vec{x})}{\mathsf{Instr}} \\
\mathsf{InstrKV} & \eqdef \xvar := \iread(\key)  \ \mid\  \iwrite(\key,\xvar) \\
\mathsf{InstrSQL} & \eqdef  \iselect{\vec{c_1}}{\xvar}{\tab}{\phi(\vec{c_2})} \ \mid\ \\
& \hspace{5mm} \iinsert{\tab}{\vec{x}} \ \mid\ \\
& \hspace{5mm} \idelete{\tab}{\phi(\vec{c})} \ \mid\ \\
& \hspace{5mm} \iupdate{\tab}{\vec{c_1}=\vec{x}}{\phi(\vec{c_2})} 
\end{align*}
\vspace{-6mm}
\caption{Program syntax. The set of all keys is denoted by $\Keys$, $\Vars$ denotes the set of local variables, $\Tables$ the set of table names, and $\Columns$ the set of column names.
We use $\phi$ to denote Boolean expressions, and $e$ to denote expressions interpreted as values. We use $\vec{\cdot}$ to denote vectors of elements.}
\label{fig:syntax}
\vspace{-4mm}
\end{figure}

Figure~\ref{fig:syntax} lists the definition of two simple programming languages that we use to represent applications running on top of Key-Value or SQL stores, respectively. A program is a set of \emph{sessions} running in parallel, each session being composed of a sequence of \emph{transactions}. Each transaction is delimited by $\ibegin$ and $\icommit$ instructions\footnote{For simplicity, we assume that all the transactions in the program commit. Aborted transactions can be ignored when reasoning about safety because their effects should be invisible to other transactions.}, and its body contains instructions that access the store, and manipulate a set of local variables $\Vars$ ranged over using $\xvar$, $\yvar$, $\ldots$.

In case of a program running on top of a Key-Value store, the instructions can be reading the value of a key and storing it to a local variable $\xvar$ ($\xvar := \iread(\key)$) , writing the value of a local variable $\xvar$ to a key ($\iwrite(\key,\xvar)$), or an assignment to a local variable $\xvar$. The set of values of keys or local variables is denoted by $\Vals$. Assignments to local variables use expressions interpreted as values whose syntax is left unspecified. Each of these instructions can be guarded by a Boolean condition $\phi(\vec{x})$ over a set of local variables $\vec{x}$ (their syntax is not important). Other constructs like $\mathtt{while}$ loops can be defined in a similar way. Let $\KVProgs$ denote the set of programs where a transaction body can contain only such instructions.

For programs running on top of SQL stores, the instructions include simplified
versions of standard SQL instructions and assignments to local variables. These
programs run in the context of a \emph{database schema} which is a (partial)
function $\DBschema:\Tables\rightharpoonup 2^\Columns$ mapping table names in
$\Tables$ to sets of column names in $\Columns$. The SQL store is an
\emph{instance} of a database schema $\DBschema$, i.e., a function $\DBinst:
\mathsf{dom}(\DBschema)\rightarrow 2^{\Rows}$ mapping each table $\tab$ in the
domain of $\DBschema$ to a set of \emph{rows} of $\tab$, i.e., functions $r:\DBschema(\tab)\rightarrow\Vals$. We use $\Rows$ to denote the set of all rows.
The $\mathtt{SELECT}$ instruction retrieves the columns $\vec{c_1}$ from the set of rows of $\tab$ that satisfy $\phi(\vec{c_2})$ ($\vec{c_2}$ denotes the set of columns used in this Boolean expression), and stores them into a variable $\xvar$. $\mathtt{INSERT}$ adds a new row to $\tab$ with values $\vec{x}$, and $\mathtt{DELETE}$ deletes all rows from $\tab$ that satisfy a condition $\phi(\vec{c})$. The $\mathtt{UPDATE}$ instruction assigns the columns $\vec{c_1}$ of all rows of $\tab$ that satisfy $\phi(\vec{c_2})$ with values in $\vec{x}$.
Let $\SQLProgs$ denote the set of programs where a transaction body can contain only such instructions.

\section{Isolation Levels for Key-Value Stores}
\label{sec:ax-kv}

%
%
%

We present the axiomatic framework introduced in~\cite{DBLP:journals/pacmpl/BiswasE19} for defining isolation levels\footnote{Isolation levels are called consistency models in~\cite{DBLP:journals/pacmpl/BiswasE19}.} in Key-Value stores. Isolation levels are defined as logical constraints, called \emph{axioms}, over \emph{histories}, which are an abstract representation of the interaction between a program and the store in a concrete execution. 


\subsection{Histories}

Programs interact with a Key-Value store by issuing transactions formed of $\textsf{read}$ and $\textsf{write}$ instructions. The effect of executing one such instruction is represented using an \emph{operation}, which is an element of the set
\begin{align*}
 \Op=\set{\rd[\id]{\key}{\val},\wrt[\id]{\key}{\val}: \id\in\OId, \key\in\Keys, \val\in \Val}
\end{align*} 
where $\rd[\id]{\key}{\val}$ (resp., $\wrt[\id]{\key}{\val}$) corresponds to reading a value $\val$ from a key $\key$ (resp., writing $\val$ to $\key$). Each operation is associated with an identifier $\id$ from an arbitrary set $\OId$. We omit operation identifiers when they are not important.

\begin{definition}
 A \emph{transaction log} $\tup{\tr,O, \po}$ is a transaction identifier $\tr$ and a finite set of operations $O$ along with a strict total order $\po$ on $O$, called \emph{program order}.
\end{definition}

The program order $\po$ represents the order between instructions in the body of a transaction. We assume that each transaction log is well-formed in the sense that if a read of a key $k$ is preceded by a write to $\key$ in $\po$, then it should return the value written by the last write to $\key$ before the read (w.r.t. $\po$). This property is implicit in the definition of every isolation level that we are aware of. For simplicity, we may use the term \emph{transaction} instead of transaction log. The  set of all transaction logs is denoted by $\mathsf{Tlogs}$.


The set of read operations $\rd{\key}{\_}$ in a transaction log $\tr$ that are \emph{not} preceded by a write to $\key$ in $\po$ is denoted by $\readOp{\tr}$. As mentioned above, the other read operations take their values from writes in the same transaction and their behavior is independent of other transactions. Also, the set of write operations $\wrt{\key}{\_}$ in $\tr$ that are \emph{not} followed by other writes to $\key$ in $\po$ is denoted by $\writeOp{\tr}$. If a transaction contains multiple writes to the same key, then only the last one (w.r.t. $\po$) can be visible to other transactions (w.r.t. any isolation level that we are aware of). The extension to sets of transaction logs is defined as usual. 
Also, we say that a transaction log $\tr$ \emph{writes} a key $\key$, denoted by $\writeVar{\tr}{\key}$, when $\wrt[\id]{\key}{\val}\in \writeOp{\tr}$ for some $\id$ and $\val$. 

A \emph{history} contains a set of transaction logs (with distinct identifiers) ordered by a (partial) \emph{session order} $\so$ that represents the order between transactions in the same session\footnote{In the context of our programming language, $\so$ would be a union of total orders. This constraint is not important for defining isolation levels.}. It also includes a
\emph{write-read} relation (also called read-from) that ``justifies'' read values by associating each read to a transaction that wrote the value returned by the read.

\begin{definition}
 A \emph{history} $\tup{T, \so, \wro}$ is a set of transaction logs $T$ along with a strict partial \emph{session order} $\so$, and a 
\emph{write-read} relation $\wro\subseteq T\times \readOp{T}$ such that
the inverse of $\wro$ is a total function, and if $(\tr,\rd{\key}{\val})\in\wro$, then $\wrt{\key}{\val}\in\tr$, and $\so\cup\wro$ is acyclic.
\end{definition}



To simplify the technical exposition, we assume that every history includes a distinguished transaction log writing the initial values of all keys. This transaction log precedes all the other transaction logs in $\so$. We use $\hist$, $\hist_1$, $\hist_2$, $\ldots$ to range over histories. The set of transaction logs $T$ in a history $\hist=\tup{T, \so, \wro}$ is denoted by $\tlogs{\hist}$.

For a key $\key$, $\wro[\key]$ denotes the restriction of $\wro$ to reads of $\key$, \ie, $\wro[\key]=\wro\cap (T\times \{\rd{\key}{\val}\mid \val\in \Val\})$. Moreover, we extend the relations $\wro$ and $\wro[\key]$ to pairs of transactions by $\tup{\tr_1,\tr_2}\in \wro$, resp., $\tup{\tr_1,\tr_2}\in \wro[\key]$, iff there exists a read operation $\rd{\key}{\val}\in \readOp{\tr_2}$ such that $\tup{\tr_1,\rd{\key}{\val}}\in \wro$, resp., $\tup{\tr_1,\rd{\key}{\val}}\in \wro[\key]$. 
We say that the transaction log $\tr_1$ is \emph{read} by the transaction log $\tr_2$ when $\tup{\tr_1,\tr_2}\in \wro$. 

\subsection{Axiomatic Framework}


\tikzset{transaction state/.style={draw=black!0}}

 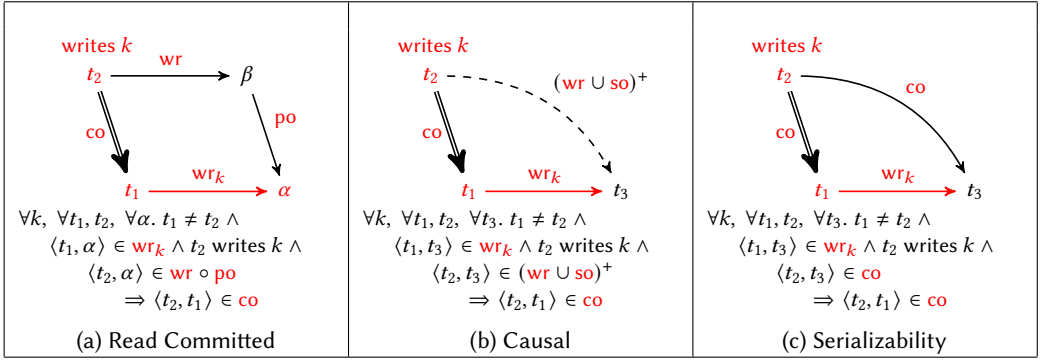
\begin{figure*}
   \resizebox{\textwidth}{!}{
   \footnotesize
  \begin{tabular}{|c|c|c|}
   \hline &  & \\
   \begin{subfigure}[t]{.3\textwidth}
    \centering
    \begin{tikzpicture}[->,>=stealth',shorten >=1pt,auto,node distance=1cm,
      semithick, transform shape]
     \node[transaction state, text=red] at (0,0)       (t_1)           {$\tr_1$};
     \node[transaction state, text=red, label={above:\textcolor{red}{$\writeVar{ }{\key}$}}] at (-0.5,1.5) (t_2) {$\tr_2$};
     \node[transaction state, text=red] at (2,0)       (o_1)           {$\alpha$};
     \node[transaction state] at (1.5,1.5) (o_2) {$\beta$};
     \path (t_1) edge[red] node {$\wro[\key]$} (o_1);
     \path (t_2) edge node {$\wro$} (o_2);
     \path (o_2) edge node {$\po$} (o_1);
     \path (t_2) edge[left,double] node {$\co$} (t_1);
    \end{tikzpicture}
    \parbox{\textwidth}{
     $\forall \key,\ \forall \tr_1, \tr_2,\ \forall \alpha.\ \tr_1\neq \tr_2\ \land$
     
     \hspace{4mm}$\tup{\tr_1,\alpha}\in \wro[\key] \land \writeVar{\tr_2}{\key}\ \land$ 
     
     \hspace{9mm}$\tup{\tr_2,\alpha}\in\wro\circ\po$
     
     \hspace{14mm}$\implies \tup{\tr_2,\tr_1}\in\co$
    }
    
    \caption{$\mathsf{Read\ Committed}$}
    \label{lock_rc_def}
   \end{subfigure}
   
%
%
%
%
%
   
   &
   
   \begin{subfigure}[t]{.3\textwidth}
    \centering
    \begin{tikzpicture}[->,>=stealth',shorten >=1pt,auto,node distance=4cm,
      semithick, transform shape]
     \node[transaction state, text=red] at (0,0)       (t_1)           {$\tr_1$};
     \node[transaction state] at (2,0)       (t_3)           {$\tr_3$};
     \node[transaction state, text=red,label={above:\textcolor{red}{$\writeVar{ }{\key}$}}] at (-.5,1.5) (t_2) {$\tr_2$};
     \path (t_1) edge[red] node {$\wro[\key]$} (t_3);
     \path (t_2) edge[dashed, bend left] node {$(\wro \cup \so)^+$} (t_3);
     \path (t_2) edge[left,double] node {$\co$} (t_1);
    \end{tikzpicture}
    \parbox{\textwidth}{
     $\forall \key,\ \forall \tr_1, \tr_2,\ \forall \tr_3.\ \tr_1\neq \tr_2\ \land$
     
     \hspace{4mm}$\tup{\tr_1,\tr_3}\in \wro[\key] \land \writeVar{\tr_2}{\key}\ \land$ 
     
     \hspace{9mm}$\tup{\tr_2,\tr_3}\in(\wro\cup\so)^+$
     
     \hspace{14mm}$\implies \tup{\tr_2,\tr_1}\in\co$
    }
    
    \caption{$\mathsf{Causal}$}
    \label{cc_def}
   \end{subfigure}

          &     
   \begin{subfigure}[t]{.3\textwidth}
    \centering
    \begin{tikzpicture}[->,>=stealth',shorten >=1pt,auto,node distance=4cm,
      semithick, transform shape]
     \node[transaction state, text=red] at (0,0)       (t_1)           {$\tr_1$};
     \node[transaction state] at (2,0)       (t_3)           {$\tr_3$};
     \node[transaction state, text=red, label={above:\textcolor{red}{$\writeVar{ }{\key}$}}] at (-.5,1.5) (t_2) {$\tr_2$};
     \path (t_1) edge[red] node {$\wro[\key]$} (t_3);
     \path (t_2) edge[bend left] node {$\CO$} (t_3);
     \path (t_2) edge[left,double] node {$\co$} (t_1);
    \end{tikzpicture}
    \parbox{\textwidth}{
     $\forall \key,\ \forall \tr_1, \tr_2,\ \forall \tr_3.\ \tr_1\neq \tr_2\ \land$
     
     \hspace{4mm}$\tup{\tr_1,\tr_3}\in \wro[\key] \land \writeVar{\tr_2}{\key}\ \land$ 
     
     \hspace{9mm}$\tup{\tr_2,\tr_3}\in\co$
     
     \hspace{14mm}$\implies \tup{\tr_2,\tr_1}\in\co$
    }
    
    \caption{$\mathsf{Serializability}$}
    \label{ser_def}
   \end{subfigure}
   \\ \hline
  \end{tabular}
  }
  \vspace{-3mm}
  \caption{Axioms defining isolations levels. The reflexive and transitive, resp., transitive, closure of a relation $rel$ is denoted by $rel^*$, resp., $rel^+$. Also, $\circ$ denotes the composition of two relations, i.e., $rel_1 \circ rel_2 = \{\tup{a, b} | \exists c. \tup{a, c} \in rel_1 \land \tup{c, b} \in rel_2\}$.}
  \label{consistency_defs}
  \vspace{-2mm}
 \end{figure*}


A history is said to satisfy a certain isolation level if there exists a strict total order $\co$ on its transaction logs, called \emph{commit order}, which extends the write-read relation and the session order, and which satisfies certain properties. These properties, called \emph{axioms}, relate the commit order with the session-order and the write-read relation in the history. 
They are defined as 
first-order formulas\footnote{These formulas are interpreted on tuples $\tup{\hist,\co}$ of a history $\hist$ and a commit order $\co$ on the transactions in $\hist$ as usual.} of the following form:
\begin{align}
  & \forall \key,\ \forall \tr_1\neq \tr_2,\ \forall \alpha.\ \nonumber\\
  & \hspace{3mm}  \tup{\tr_1,\alpha}\in \wro[\key] \land \writeVar{\tr_2}{\key} \land \phi(\tr_2,\alpha) \implies \tup{\tr_2,\tr_1}\in\co \label{eq:axiom}
\end{align}
where $\phi$ is a property relating $\tr_2$ and $\alpha$ (i.e., the read or the transaction reading from $\tr_1$) that varies from one axiom to another. Intuitively, this axiom schema states the following: in order for $\alpha$ to read specifically $t_1$'s write on $x$, it must be the case that every $t_2$ that also writes $x$ and satisfies $\phi(t_2,\alpha)$ was committed before $t_1$. 
The property $\phi$ relates $\tr_2$ and $\alpha$ using the relations in a history and the commit order. 
Figure~\ref{consistency_defs} shows the axioms defining three isolation levels: Read Committed, Causal Consistency, and Serializability (see~\cite{DBLP:journals/pacmpl/BiswasE19} for axioms defining Read Atomic, Prefix, and Snapshot Isolation).

For instance, $\mathsf{Read\ Committed}$~\cite{DBLP:conf/sigmod/BerensonBGMOO95} requires that every read returns a value written in a committed transaction, and also, that the reads in the same transaction are ``monotonic'', i.e., they do not return values that are older, w.r.t. the commit order, than values read in the past.
While the first condition holds for every history (because of the surjectivity of $\wro$), the second condition is expressed by the axiom $\mathsf{Read\ Committed}$ in Figure~\ref{lock_rc_def}, which states that for any transaction $\tr_1$ writing a key $\key$ that is read in a transaction $\tr$, the set of transactions $\tr_2$ writing $\key$ and read previously in the same transaction (these reads may concern other keys) must precede $\tr_1$ in commit order. 
For instance, Figure~\ref{rc_example:1} shows a history and a (partial) commit order that does not satisfy this axiom because $\rd{\key_1}{1}$ returns the value written in a transaction ``older'' than the transaction read in the previous $\rd{\key_2}{2}$. 

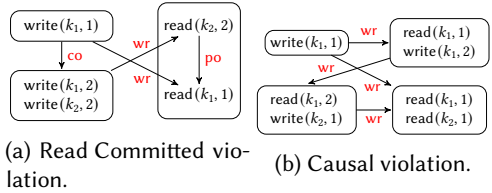
\begin{figure}
  
   \centering
   \begin{subfigure}{.23\textwidth}
  \resizebox{\textwidth}{!}{
\begin{tikzpicture}[->,>=stealth',shorten >=1pt,auto,node distance=3cm,
    semithick, transform shape]
    \node[draw, rounded corners=2mm] (t1) at (0, 0) {\begin{tabular}{l} $\wrt{\key_1}{1}$ \end{tabular}};
   \node[draw, rounded corners=2mm,outer sep=0] (t2) at (0, -1.5) {\begin{tabular}{l} $\wrt{\key_1}{2}$ \\ $\wrt{\key_2}{2}$\end{tabular}};
   \node[draw, rounded corners=2mm, minimum width=1.8cm, minimum height=2.5cm] (t3) at (3, -0.75) {};
   \node[style={inner sep=0,outer sep=0}] (t3_1) at (3, 0) {\begin{tabular}{l} $\rd{\key_2}{2}$ \end{tabular}};
   \node[style={inner sep=0,outer sep=0}] (t3_2) at (3, -1.5) {\begin{tabular}{l} $\rd{\key_1}{1}$ \end{tabular}};
   \path (t1) edge node {$\co$} (t2);
   \path (t3_1) edge node {$\po$} (t3_2);
   \path (t1) edge[below] node[yshift=-4,xshift=4] {$\wro$} (t3_2);
   \path (t2) edge node[yshift=-2,xshift=7] {$\wro$} (t3_1);
  \end{tikzpicture}  
    }
    \caption{$\mathsf{Read\ Committed}$ violation.}
    \label{rc_example:1}
\end{subfigure}
\begin{subfigure}{.23\textwidth}
\resizebox{\textwidth}{!}{
\begin{tikzpicture}[->,>=stealth',shorten >=1pt,auto,node distance=3cm,
 semithick, transform shape]
 \node[draw, rounded corners=2mm,outer sep=0] (t1) at (0, 1.5) {$\wrt{\key_1}{1}$};
\node[draw, rounded corners=2mm,outer sep=0] (t2) at (3, 1.5) {\begin{tabular}{l} $\rd{\key_1}{1}$ \\ $\wrt{\key_1}{2}$ \end{tabular}};
\node[draw, rounded corners=2mm,outer sep=0] (t3) at (3, 0) {\begin{tabular}{l} $\rd{\key_1}{1}$ \\ $\rd{\key_2}{1}$ \end{tabular}};
\node[draw, rounded corners=2mm,outer sep=0] (t4) at (0, 0) {\begin{tabular}{l} $\rd{\key_1}{2}$ \\ $\wrt{\key_2}{1}$\end{tabular}};

\path (t1) edge[above] node[yshift=0,xshift=0] {$\wro$} (t2);

\path (t1) edge[below] node[yshift=-5,xshift=7] {$\wro$} (t3);

\path (t2) edge[above] node[yshift=-6,xshift=-14] {$\wro$} (0,0.58);

\path (t4) edge[below] node[yshift=0,xshift=0] {$\wro$} (t3);
\end{tikzpicture}  
}
 \caption{$\mathsf{Causal}$ violation.}
 \label{cc_example:1}
\end{subfigure}
\vspace{-3mm}
  \caption{Histories used to explain the axioms in Figure~\ref{consistency_defs}.}
  \label{counter_example:1}
\vspace{-3mm}
\end{figure}

The axiom defining $\mathsf{Causal}$ Consistency~\cite{DBLP:journals/cacm/Lamport78} states that for any transaction $\tr_1$ writing a key $\key$ that is read in a transaction $\tr_3$, the set of $(\wro\cup \so)^+$ predecessors of $\tr_3$ writing $\key$ must precede $\tr_1$ in commit order ($(\wro\cup \so)^+$ is usually called the \emph{causal} order). A violation of this axiom can be found in Figure~\ref{cc_example:1}: the transaction $\tr_2$ writing 2 to $\key_1$ is a $(\wro\cup \so)^+$ predecessor of the transaction $\tr_3$ reading 1 from $\key_1$ because the transaction $\tr_4$, writing 1 to $\key_2$, reads $\key_1$ from $\tr_2$ and $\tr_3$ reads $\key_2$ from $\tr_4$. This implies that $\tr_2$ should precede in commit order the transaction $\tr_1$ writing 1 to $\key_1$, which again, is inconsistent with the write-read relation ($\tr_2$ reads from $\tr_1$).

Finally, $\mathsf{Serializability}$~\cite{DBLP:journals/jacm/Papadimitriou79b} requires that for any transaction $\tr_1$ writing to a key $\key$ that is read in a transaction $\tr_3$, the set of $\co$ predecessors of $\tr_3$ writing $\key$ must precede $\tr_1$ in commit order. This ensures that each transaction observes the effects of all the $\co$ predecessors. 


\begin{definition}
For an isolation level $I$ defined by a set\footnote{Isolation levels like Snapshot Isolation require more than one axiom.} of axioms $X$, a history $\hist=\tup{T, \so, \wro}$ \emph{satisfies} $I$ iff there is a strict total order $\co$ s.t. $\wro\cup\so\subseteq \co$ and $\tup{h,\co}$ satisfies $X$.
 \label{axiom-criterion}
\end{definition}

\section{Operational Semantics for $\KVProgs$}
\label{sec:op-kv}

We define a small-step operational semantics for Key-Value store programs, which is parametrized by an isolation level $I$. Transactions are executed \emph{serially} one after another, and the values returned by $\rdo$ operations are decided using the axiomatic definition of $I$. The semantics maintains a history of previously executed operations, and the value returned by a $\rdo$ is chosen non-deterministically as long as extending the current history with the corresponding write-read dependency satisfies the axioms of $I$. 
We show that this semantics is sound and complete for any natural isolation level $I$, i.e., it generates precisely the same set of histories as a \emph{baseline} semantics where transactions can interleave arbitrarily and the $\rdo$ operations can return arbitrary values as long as they can be proved to be correct at the end of the execution.

\subsection{Definition of the Operational Semantics}

\tikzset{
  keep name/.style={
    prefix after command={
      \pgfextra{\let\fixname\tikzlastnode}
    }
  },
  partialbox/.style={
    keep name,
    append after command={
  ([xshift=#1]\fixname.north west) -- 
  (\fixname.north west) -- 
  (\fixname.south west) -- 
  ([xshift=#1]\fixname.south west)
  ([xshift=-#1]\fixname.north east) -- 
  (\fixname.north east) -- 
  (\fixname.south east) -- 
  ([xshift=-#1]\fixname.south east)
    }
  },
  partialbox/.default=15pt
}

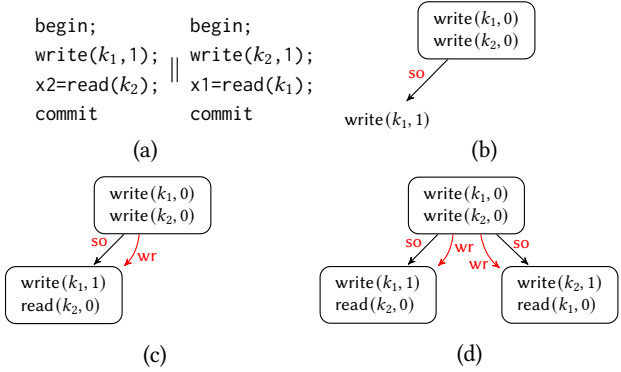
\begin{figure}
\begin{minipage}{2.2cm}
\begin{lstlisting}[xleftmargin=5mm,basicstyle=\ttfamily\footnotesize,escapeinside={(*}{*)}]
begin;
write((*$\key_1$*),1);
x2=read((*$\key_2$*));
commit
\end{lstlisting}
\end{minipage}
\begin{minipage}{1mm}
||
\end{minipage}
\hspace{-5mm}
\begin{minipage}{2.2cm}
\begin{lstlisting}[xleftmargin=5mm,basicstyle=\ttfamily\footnotesize,escapeinside={(*}{*)}]
begin;
write((*$\key_2$*),1);
x1=read((*$\key_1$*));
commit
\end{lstlisting}
\end{minipage}
\begin{minipage}{4.1cm}
  \resizebox{.7\textwidth}{!}{
   \begin{tikzpicture}[->,>=stealth',shorten >=1pt,auto,node distance=4cm,
     semithick, transform shape]
    \node[draw, rounded corners=2mm] (t0) at (1.7, 1.7) {\begin{tabular}{l} $\wrt{\key_1}{0}$ \\ $\wrt{\key_2}{0}$ \end{tabular}};
    \node(s11) at (0, 0) {\begin{tabular}{l} $\wrt{\key_1}{1}$ \end{tabular}};
    \path (t0) edge[above] node[pos=0.6, xshift=-3] {$\so$} (s11);
   \end{tikzpicture}  
  }
\end{minipage}

{\small (a)} \hspace{4cm} {\small (b)}

\medskip
\begin{minipage}{4.1cm}
  \resizebox{.7\textwidth}{!}{
   \begin{tikzpicture}[->,>=stealth',shorten >=1pt,auto,node distance=4cm,
     semithick, transform shape]
    \node[draw, rounded corners=2mm] (t0) at (1.7, 1.7) {\begin{tabular}{l} $\wrt{\key_1}{0}$ \\ $\wrt{\key_2}{0}$ \end{tabular}};
    \node[draw, rounded corners=2mm] (s11) at (0, 0) {\begin{tabular}{l} $\wrt{\key_1}{1}$ \\ $\rd{\key_2}{0}$ \end{tabular}};
    \path (t0) edge[above] node[pos=0.6, xshift=-3] {$\so$} (s11);
    \path (t0) edge[red, right, bend left=20] node[pos=0.7] {$\wro$} (s11);
   \end{tikzpicture}  
  }
  
  \begin{center}
  {\small (c)}
  \end{center}
\end{minipage}
\begin{minipage}{4.1cm}
  \resizebox{\textwidth}{!}{
   \begin{tikzpicture}[->,>=stealth',shorten >=1pt,auto,node distance=4cm,
     semithick, transform shape]
    \node[draw, rounded corners=2mm] (t0) at (1.7, 1.7) {\begin{tabular}{l} $\wrt{\key_1}{0}$ \\ $\wrt{\key_2}{0}$ \end{tabular}};
    \node[draw, rounded corners=2mm] (s11) at (0, 0) {\begin{tabular}{l} $\wrt{\key_1}{1}$ \\ $\rd{\key_2}{0}$ \end{tabular}};
    \node[draw, rounded corners=2mm] (s12) at (3.5, 0) {\begin{tabular}{l} $\wrt{\key_2}{1}$ \\ $\rd{\key_1}{0}$ \end{tabular}};
    \path (t0) edge[above] node[pos=0.6, xshift=-3] {$\so$} (s11);
    \path (t0) edge[right] node[pos=0.3] {$\so$} (s12);
    \path (t0) edge[red, right, bend left=20] node[pos=0.4,xshift=-1] {$\wro$} (s11);
    \path (t0) edge[red, left, bend right=20] node[pos=0.9,xshift=-1] {$\wro$} (s12);
   \end{tikzpicture}  
  }
  
    \begin{center}
  {\small (d)}
  \end{center}
\end{minipage}
\vspace{-3mm}
\caption{The $\mathsf{Causal}$ semantics on the program in (a), assuming that the transaction on the left is scheduled first. 
}
\label{fig:opsEx}
\vspace{-5mm}
\end{figure}

%

We use the program in Figure~\ref{fig:opsEx}a to give an overview of our semantics, assuming Causal Consistency. This program has two concurrent transactions whose reads can both return the initial value $0$, which is not possible under $\mathsf{Serializability}$. 

Our semantics executes transactions in their entirety one after another (without interleaving them), maintaining a history that contains all the executed operations. We assume that the transaction on the left executes first. Initially, the history contains a fictitious transaction log that writes the initial value 0 to all keys, and that will precede all the transaction logs created during the execution in session order. 

Executing a write instruction consists in simply appending the corresponding write operation to the log of the current transaction. For instance, executing the first write (and $\ibegin$) in our example results in adding a transaction log that contains a write operation (see Figure~\ref{fig:opsEx}b). The execution continues with the read instruction from the same transaction, and it cannot switch to the other transaction.

The execution of a read instruction consists in choosing non-deterministically a write-read dependency that validates $\mathsf{Causal}$ when added to the current history. In our example, executing $\iread(\key_2)$ results in adding a write-read dependency from the transaction log writing initial values, which determines the return value of the $\iread$ (see Figure~\ref{fig:opsEx}c). This choice makes the obtained history satisfy $\mathsf{Causal}$. 

The second transaction executes in a similar manner. When executing its read instruction, the chosen write-read dependency is again related to the transaction log writing initial values (see Figure~\ref{fig:opsEx}d). This choice is valid under $\mathsf{Causal}$. Since a read must not read from the preceding transaction, this semantics is able to simulate all the ``anomalies'' of a weak isolation level (this execution being an example).

Formally, the operational semantics is defined as a transition relation $\Rightarrow_I$ between \emph{configurations}, which are defined as tuples containing the following:
\begin{itemize}
	\item history $\hist$ storing the operations executed in the past, 
	\item identifier $j$ of the current session,
	\item local variable valuation $\gamma$ for the current transaction, 
	\item code $\mathsf{B}$ that remains to be executed from the current transaction, and
	\item sessions/transactions $\mathsf{P}$ that remain to be executed from the original program.
\end{itemize}

For readability, we define a program as a partial function $\mathsf{P}:\mathsf{SessId}\rightharpoonup \mathsf{Sess}$ that associates session identifiers in $\mathsf{SessId}$ with concrete code as defined in Figure~\ref{fig:syntax} (i.e., sequences of transactions). Similarly, the session order $\so$ in a history is defined as a partial function $\so:\mathsf{SessId}\rightharpoonup \mathsf{Tlogs}^*$ that associates session identifiers with sequences of transaction logs. Two transaction logs are ordered by $\so$ if one occurs before the other in some sequence $\so(j)$ with 
$j\in \mathsf{SessId}$.

Before presenting the definition of $\Rightarrow_I$, we introduce some notation. Let $\hist$ be a history that contains a representation of $\so$ as above. We use $\hist\oplus_j \tup{\tr,O,\po}$ to denote a history where $\tup{\tr,O,\po}$ is appended to $\so(j)$. 
Also, for an operation $o$, $\hist\oplus_j o$ is the history obtained from $\hist$ by adding $o$ to the last transaction log in $\so(j)$ and as a last operation in the program order of this log (i.e.,  if $\so(j)=\sigma; \tup{t,O,\po}$, then the session order $\so'$ of $\hist\oplus_j o$ is defined by $\so'(k)=\so(k)$ for all $k\neq j$ and $\so(j) =\sigma; \tup{t,O\cup{o},\po\cup \{(o',o): o'\in O\}}$). Finally, for a history $\hist = \tup{T, \so, \wro}$, $\hist\oplus\wro(\tr,o)$ is the history obtained from $\hist$ by adding $(\tr,o)$ to the write-read relation.

\begin{figure} [t]
\small
  \centering
  \begin{mathpar}
    \inferrule[spawn]{\tr \mbox{ fresh}\quad \mathsf{P}(j) = \ibegin; \mathsf{Body}; \icommit; \mathsf{S}}{
      \hist,\_,\_,\epsilon,\mathsf{P}
      \Rightarrow_I
      \hist \oplus_j \tup{\tr,\emptyset,\emptyset},j,\emptyset,\mathsf{Body},\mathsf{P}[j\mapsto \mathsf{S}]
    } 

    \inferrule[if-true]{\varphi(\vec{x})[x\mapsto \gamma(x): x\in\vec{x}]\mbox{ true}}{
      \hist,j,\gamma,\iif{\phi(\vec{x})}{\mathsf{Instr}};\mathsf{B}, \mathsf{P}
      \Rightarrow_I
      \hist,j,\gamma,\mathsf{Instr};\mathsf{B},\mathsf{P}
    } 

    \inferrule[if-false]{\varphi(\vec{x})[x\mapsto \gamma(x): x\in\vec{x}]\mbox{ false}}{
      \hist,j,\gamma,\iif{\phi(\vec{x})}{\mathsf{Instr}};\mathsf{B}, \mathsf{P}
      \Rightarrow_I
      \hist,j,\gamma,\mathsf{B},\mathsf{P}
    } 

    \inferrule[write]{v = \gamma(x)\quad \id\mbox{ fresh}}{
      \hist,j,\gamma, \iwrite(\key,\xvar);\mathsf{B}, \mathsf{P}
      \Rightarrow_I
      \hist \oplus_j \wrt[\id]{\key}{\val},j,\gamma,\mathsf{B},\mathsf{P}
    } 

    \inferrule[read-local]{
    \wrt{\key}{\val}\mbox{ is the last write on $\key$ in $\tr$ w.r.t. $\po$}\\
    \id\mbox{ fresh }
    }{
      \hist,j,\gamma, \xvar := \iread(\key);\mathsf{B}, \mathsf{P}
      \Rightarrow_I
      \hist \oplus_j \rd[\id]{\key}{\val},j,\gamma[\xvar\mapsto \val],\mathsf{B},\mathsf{P}
    } 

    \inferrule[read-extern]{
    \hist=(T,\so,\wro) \\
    \tr \mbox{ is the id of the last transaction log in $\so(j)$} \\
    \wrt{\key}{\val}\in\writeOp{\tr'}\mbox{ with $\tr'\in T$ and $\tr'\neq \tr$} \\
    \id\mbox{ fresh }\\
    \hist' = (\hist \oplus_j \rd[\id]{\key}{\val}) \oplus \wro(\tr',\rd[i]{\key}{\val}) \\
    \hist' \mbox{ satisfies } I }{
      \hist,j,\gamma, \xvar := \iread(\key);\mathsf{B}, \mathsf{P}
      \Rightarrow_I 
      \hist',j,\gamma[\xvar\mapsto \val], \mathsf{B}, \mathsf{P}
    } 
    
  \end{mathpar}
 \vspace{-4mm}
  \caption{Operational semantics for $\KVProgs$ programs under isolation level $I$. For a function $f:A\rightharpoonup B$, $f[a\mapsto b]$ denotes the function $f':A\rightharpoonup B$ defined by $f'(c) = f(c)$, for every $c\neq a$ in the domain of $f$, and $f'(a)=b$.}
  \label{fig:op:sem}
 \vspace{-4mm}
\end{figure}

Figure~\ref{fig:op:sem} lists the rules defining $\Rightarrow_I$. The \textsc{spawn} rule starts a new transaction, provided that there is no other live transaction ($\mathsf{B}=\epsilon$). It adds an empty transaction log to the history and schedules the body of the transaction. \textsc{if-true} and \textsc{if-false} check the truth value of a Boolean condition of an $\mathtt{if}$ conditional. \textsc{write} corresponds to a write instruction and consists in simply adding a write operation to the current history. \textsc{read-local} and \textsc{read-extern} concern read instructions. \textsc{read-local} handles the case where the read follows a write on the same key $k$ in the same transaction: the read returns the value written by the last write on $k$ in the current transaction. Otherwise, \textsc{read-extern} corresponds to reading a value written in another transaction $\tr'$ ($\tr$ is the id of the log of the current transaction). The transaction $\tr'$ is chosen non-deterministically as long as extending the current history with the write-read dependency associated to this choice leads to a history that still satisfies $I$.

An \emph{initial} configuration for program $\prog$ contains the program $\prog$ along with a history $\hist=\tup{\{\tr_0\},\emptyset,\emptyset}$, where $\tr_0$ is a transaction log containing only writes that write the initial values of all keys, and empty current transaction code ($\mathsf{B}=\epsilon$). 
An execution of a program $\prog$ under an isolation level $I$ is a sequence of configurations $c_0 c_1\ldots c_n$ where $c_0$ is an initial configuration for $\prog$, and $c_m\Rightarrow_I c_{m+1}$, for every $0\leq m < n$. We say that $c_n$ is \emph{$I$-reachable} from $c_0$.
The history of such an execution is the history $\hist$ in the last configuration $c_n$. 
A configuration is called \emph{final} if it contains the empty program ($\prog=\emptyset$).
Let $\histOf[I]{\prog}$ denote the set of all histories of an execution of $\prog$ under $I$ that ends in a final configuration.

%
%


\subsection{Correctness of the Operational Semantics}

\begin{figure} [t]
\small
  \centering
  \begin{mathpar}
    \inferrule[spawn*]{\tr \mbox{ fresh}\quad \mathsf{P}(j) = \ibegin; \mathsf{Body}; \icommit; \mathsf{S} \quad \vec{\mathsf{B}}(j) = \epsilon}{
      \hist,\vec{\gamma},\vec{\mathsf{B}},\mathsf{P}
      \Rightarrow
      \hist \oplus_j \tup{\tr,\emptyset,\emptyset},\vec{\gamma}[j\mapsto \emptyset],\vec{\mathsf{B}}[j\mapsto \mathsf{Body}],\mathsf{P}[j\mapsto \mathsf{S}]
    } 

%
%
%

    \inferrule[read-extern*]{
    \vec{\mathsf{B}}(j) = \xvar := \iread(\key);\mathsf{B} \\
    \hist=(T,\so,\wro) \\
    \tr \mbox{ is the id of the last transaction log in $\so(j)$} \\
    \wrt{\key}{\val}\in\writeOp{\tr'}\mbox{ with $\tr'\in \transC{\hist,\vec{\mathsf{B}}}$ and $\tr\neq \tr'$} \\
    \id\mbox{ fresh }\\
    \hist' = (\hist \oplus_j \rd[\id]{\key}{\val}) \oplus \wro(\tr',\rd[i]{\key}{\val}) }{
      \hist,\vec{\gamma},\vec{\mathsf{B}}, \mathsf{P}
      \Rightarrow
      \hist',\vec{\gamma}[(j,\xvar)\mapsto \val],\vec{\mathsf{B}}[j\mapsto \mathsf{B}],\mathsf{P}
    } 
    
  \end{mathpar}
 \vspace{-4mm}
  \caption{A baseline operational semantics for $\KVProgs$ programs. Above, $\transC{\hist,\vec{\mathsf{B}}}$ denotes the set of transaction logs in $\hist$ that excludes those corresponding to live transactions, i.e., transaction logs $\tr''\in T$ such that $\tr''$ is the last transaction log in some $\so(j')$ and $\vec{B}(j')\neq\epsilon$.}
  \label{fig:op:sem:baseline}
 \vspace{-4mm}
\end{figure}

We define the correctness of $\Rightarrow_I$ in relation to a \emph{baseline} semantics where transactions can interleave arbitrarily, and the values returned by $\rdo$ operations are only constrained to come from committed transactions. 
This semantics is represented by a transition relation $\Rightarrow$, which is defined by a set of rules that are analogous to $\Rightarrow_I$. 
Since it allows transactions to interleave, a configuration contains a history $\hist$, the sessions/transactions $\mathsf{P}$ that remain to be executed, and:
\begin{itemize}
	\item a valuation map $\vec{\gamma}$ that records local variable values in the current transaction of each session ($\vec{\gamma}$ associates identifiers of sessions that have live transactions with valuations of local variables),
	\item a map $\vec{B}$ that stores the code of each live transaction (associating session identifiers with code).
\end{itemize}
Figure~\ref{fig:op:sem:baseline} lists some rules defining $\Rightarrow$ (the others can be defined in a similar manner). \textsc{spawn*} starts a new transaction in a session $j$ provided that this session has no live transaction ($\vec{\mathsf{B}}(j) = \epsilon$). Compared to \textsc{spawn} in Figure~\ref{fig:op:sem}, this rule allows unfinished transactions in other sessions. \textsc{read-extern*} does not check conformance to $I$, but it allows a read to only return a value written in a completed (committed) transaction. In this work, we consider only isolation levels satisfying this constraint. Executions, initial and final configurations are defined as in the case of $\Rightarrow_I$. The history of an execution is still defined as the history in the last configuration. Let $\histOf[*]{\prog}$ denote the set of all histories of an execution of $\prog$ w.r.t. $\Rightarrow$ that ends in a final configuration.

Practical isolation levels satisfy a ``prefix-closure'' property saying that if the axioms of $I$ are satisfied by a pair $\tup{\hist_2,\co_2}$, then they are also satisfied by every \emph{prefix} of $\tup{\hist_2,\co_2}$. A prefix of $\tup{\hist_2,\co_2}$ contains a prefix of the sequence of transactions in $\hist_2$ when ordered according to $\co_2$, and the last transaction log in this prefix is possibly incomplete.
In general, this prefix-closure property holds for isolation levels $I$ that are defined by axioms as in (\ref{eq:axiom}), provided that the property $\phi(\tr_2,\alpha)$ is \emph{monotonic}, i.e., the set of models in the context of a pair $\tup{\hist_2,\co_2}$ is a \emph{superset} of the set of models in the context of a prefix $ \tup{\hist_1,\co_1}$ of $\tup{\hist_2,\co_2}$. For instance, the property $\phi$ in the axiom defining $\mathsf{Causal}$ is $(\tr_2,\alpha)\in (\wro \cup \so)^+$, which is clearly monotonic. In general, standard isolation levels are defined using a property $\alpha$ of the form $(\tr_2,\alpha)\in R$ where $R$ is an expression built from the relations $\po$, $\so$, $\wro$, and $\co$ using (reflexive and) transitive closure and composition of relations~\cite{DBLP:journals/pacmpl/BiswasE19}. Such properties are monotonic in general (they would not be if those expressions would use the negation/complement of a relation).  An axiom as in (\ref{eq:axiom}) is called \emph{monotonic} when the property $\phi$ is monotonic. 

The following theorem shows that $\histOf[I]{\prog}$ is precisely the set of histories under the baseline semantics, which satisfy $I$ (the validity of the reads is checked at the end of an execution), provided that the axioms of $I$ are monotonic.

 \begin{theorem}
For any isolation level $I$ defined by a set of monotonic axioms,
$
\histOf[I]{\prog} = \{ h \in \histOf[*]{\prog}: h\mbox{ satisfies }I\}.
$
 \end{theorem}
 
The $\subseteq$ direction follows mostly from the fact that $\Rightarrow_I$ is more constrained than $\Rightarrow$. For the opposite direction, given a history $\hist$ that satisfies $I$, i.e., there exists a commit order $\co$ such that $\tup{h,\co}$ satisfies the axioms of $I$, we can show that there exists an execution under $\Rightarrow_I$ with history $\hist$, where transactions  execute serially in the order defined by $\co$. The prefix closure property is used to prove that \textsc{read-extern} transitions are enabled (these transitions get executed with a prefix of $\hist$). See the supplementary material for more details.
  
It can also be shown that $\Rightarrow_I$ is \emph{deadlock-free} for every natural isolation level (e.g., Read Committed, Causal Consistency, Snapshot Isolation, and Serializability), i.e., every read can return some value satisfying the axioms of $I$ at the time when it is executed (independently of previous choices). 

\section{Compiling SQL to Key-Value API}
\label{sec:SQL-to-KV}

We define an operational semantics for SQL programs (in $\SQLProgs$) based on a compiler that rewrites SQL queries to Key-Value $\iread$ and $\iwrite$ instructions. For presentation reasons, we use an intermediate representation where each table of a database instance is represented using a \emph{set} variable that stores values of the primary key\footnote{For simplicity, we assume that primary keys correspond to a single column in the table.} (identifying uniquely the rows in the table) and a set of key-value pairs, one for each cell in the table. In a second step, we define a rewriting of the API used to manipulate set variables into Key-Value $\iread$ and $\iwrite$ instructions.

\paragraph{Intermediate Representation}

Let $\DBschema:\Tables\rightharpoonup 2^\Columns$ be a database schema (recall that $\Tables$ and $\Columns$ are the set of table names and column names, resp.). For each table $\tab$, let $\tab.\pkey$ be the name of the primary key column. We represent an instance $\DBinst: \mathsf{dom}(\DBschema)\rightarrow 2^{\Rows}$ using:
\begin{itemize}
	\item for each table $\tab$, a set variable $\tab$ (with the same name) that contains the primary key value $r(\tab.\pkey)$ of every row $r\in \DBinst(\tab)$, 
	\item for each row $r\in \DBinst(\tab)$ with primary key value $\pkeyVal = r(\tab.\pkey)$, and each column $c\in \DBschema(\tab)$, a key $\tab.\pkeyVal.c$ associated with the value $r(c)$.
\end{itemize}

\begin{figure}[t]
{\footnotesize
\begin{minipage}[t]{3cm}
\setlength{\tabcolsep}{3pt}
\begin{center}
Table:
\end{center}

\begin{tabular}{|c|c|c|} \hline
\multicolumn{3}{|c|}{A} \\ \hline\hline
Id & Name & City \\ \hline
1 & Alice & Paris \\ \hline
2 & Bob & Bangalore \\ \hline
3 & Charles & Bucharest \\ \hline
\end{tabular}
\end{minipage}
\begin{minipage}[t]{5cm}
\begin{center}
Intermediate representation:
\end{center}
\setlength{\tabcolsep}{1pt}
\begin{tabular}{lll} 
\multicolumn{3}{l}{A = \{ 1, 2, 3 \}} \\
\\[1mm]
A.1.Id: 1, & A.1.Name: Alice, & A.1.City: Paris \\
A.2.Id: 2, & A.2.Name: Bob, & A.2.City: Bangalore \\
A.3.Id: 3, & A.3.Name: Charles, & A.3.City: Bucharest
\end{tabular}
\end{minipage}}
 \vspace{-2mm}
\caption{Representing tables with set variables and key-value pairs. We write a key-value pair as key:value.}
\label{fig:sql-example}
 \vspace{-2mm}
\end{figure}

\begin{example}

The table A on the left of Figure~\ref{fig:sql-example}, where the primary key is defined by the Id column, is represented using a set variable A storing the set of values in the column Id, and one key-value pair for each cell in the table.

\end{example}

\begin{figure}[t]
\small
\begin{flushleft}
\begin{minipage}{6cm}
\begin{flushleft}
\texttt{SELECT}/\texttt{DELETE}/\texttt{UPDATE}
\end{flushleft}
\vspace{-2mm}
\begin{lstlisting}[xleftmargin=5mm,language=MyLang,escapeinside={(*}{*)}]
rows := elements(tab)
for ( let pkeyVal of rows ) {
   for ( let c of (*$\vec{c_2}$*) ) {
      val[c] := read(tab.pkeyVal.c)
   if ( (*$\phi[\texttt{c}\mapsto \texttt{val[c]}: \texttt{c}\in\vec{c_2}]$*) true )
      // (*$\iselect{\vec{c_1}}{\xvar}{\tab}{\phi(\vec{c_2})}$*)
      for ( let c of (*$\vec{c_1}$*) )
         out[c] := read(tab.pkeyVal.c)
      x := x (*$\cup$*) out
      // (*$\idelete{\tab}{\phi(\vec{c_2})}$*)
      remove(tab, pkeyVal);
      // (*$\iupdate{\tab}{\vec{c_1}=\vec{x}}{\phi(\vec{c_2})}$*)
      for ( let c of (*$\vec{c_1}$*) )
         write( tab.pkeyVal.c, (*$\gamma$*)((*$\vec{x}$*)[c]) )
\end{lstlisting}
\end{minipage}%
\begin{minipage}{2cm}
	~
\end{minipage}%
\begin{minipage}{3cm}

\begin{flushleft}
$\iinsert{\tab}{\vec{x}}$
\end{flushleft}
\vspace{-2mm}
\begin{lstlisting}[xleftmargin=5mm,language=MyLang,escapeinside={(*}{*)}]
pkeyVal := (*$\gamma$*)((*$\vec{x}$*)[0])
if ( add(tab,pkeyVal) ) {
   for ( let c of (*$\DBschema(\texttt{tab})$*) ) {
      write( tab.pkeyVal.c, (*$\gamma$*)((*$\vec{x}$*)[c]) )
\end{lstlisting}
\end{minipage}
\end{flushleft}
    \caption{Compiling SQL queries to the intermediate representation. Above, $\gamma$ is a valuation of local variables. Also, in the case of $\mathtt{INSERT}$, we assume that the first element of $\vec{x}$ represents the value of the primary key.}
    \label{fig:sql-ir}
\end{figure}

Figure~\ref{fig:sql-ir} lists our rewriting of SQL queries over a database instance $\DBinst$ to programs that manipulate the set variables and key-value pairs described above. This rewriting contains the minimal set of accesses to the cells of a table that are needed to implement an SQL query according to its conventional specification. To manipulate set variables, we use $\iadd$ and $\iremove$ for adding and removing elements, respectively (returning $\btrue$ or $\bfalse$ when the element is already present or deleted from the set, respectively), and $\ielements$ that returns all of the elements in the input set\footnote{$\iadd(s,e)$ and $\iremove(s,e)$ add and remove the element $e$ from $s$, respectively. $\ielements(s)$ returns the content of $s$.}. 

$\mathsf{SELECT}$, $\mathsf{DELETE}$, and $\mathsf{UPDATE}$ start by reading the contents of the set variable storing primary key values and then, for every row, the columns in $\vec{c_2}$ needed to check the Boolean condition $\phi$ (the keys corresponding to these columns). For every row satisfying this Boolean condition, $\mathsf{SELECT}$ continues by reading the keys associated to the columns that need to be returned, $\mathsf{DELETE}$ removes the primary key value associated to this row from the set $\tab$, and $\mathsf{UPDATE}$ writes to the keys corresponding to the columns that need to be updated. In the case of $\mathsf{UPDATE}$, we assume that the values of the variables in $\vec{x}$ are obtained from a valuation $\gamma$ (this valuation would be maintained by the operational semantics of the underlying Key-Value store). $\mathsf{INSERT}$ adds a new primary key value to the set variable $\tab$ (the call to $\iadd$ checks whether this value is unique) and then writes to the keys representing columns of this new row.

\paragraph{Manipulating Set Variables}

\begin{figure}[t]
\small
\begin{minipage}[t]{4.2cm}
\begin{flushleft}
$\iadd(tab,pkeyVal)$:
\end{flushleft}
\vspace{-2mm}
\begin{lstlisting}[language=MyLang,escapeinside={(*}{*)}] 
if (read((*$tab$*).has.(*$pkeyVal$*))) 
   return false;
write((*$tab$*).has.(*$pkeyVal$*),true)
return true;
\end{lstlisting}
\end{minipage}
%
%
%
\begin{minipage}[t]{4cm}
\begin{flushleft}
$\ielements(tab)$:
\end{flushleft}
\vspace{-2mm}
\begin{lstlisting}[language=MyLang,escapeinside={(*}{*)}]
ret := (*$\emptyset$*)
for ( let (*$pkeyVal$*) of (*$\Vals$*) )
   if (read((*$tab$*).has.(*$pkeyVal$*))) 
      ret := ret (*$\cup$*) {(*$pkeyVal$*)}
return ret;
\end{lstlisting}
\end{minipage}
\vspace{-4mm}
    \caption{Manipulating set variables using key-value pairs.}
    \label{fig:ir-key}
\vspace{-3mm}
\end{figure}

Based on the standard representation of a set using its characteristic function, we implement each set variable $\tab$ using a set of keys $\tab.\icontains.\pkeyVal$, one for each value $\pkeyVal\in\Vals$. These keys are associated with Boolean values, indicating whether $\pkeyVal$ is contained in $\tab$. In a concrete implementation, this set of keys need not be fixed a-priori, but can grow during the execution with every new instance of an $\mathtt{INSERT}$. Figure~\ref{fig:ir-key} lists the implementations of $\iadd$/$\ielements$, which are self-explanatory ($\iremove$ is analogous).

%
%
%


\section{Implementation}
\label{sec:impl}

We implemented \tool{}\footnote{We plan to make MonkeyDB available open-source
soon.} to support an interface common to most storage
systems. Operations can be either key-value (KV) updates (to access data as a KV map)
or SQL queries (to access data as a relational database). \tool{} supports
transactions as well; a transaction can include multiple operations. 
Figure~\ref{fig:block_dia} shows the architecture of \tool{}. 
A client can connect to \tool{} over a TCP connection, as is
standard for SQL databases\footnote{We support the MySQL client-server
protocol using \url{https://github.com/jonhoo/msql-srv}.}. 
This
offers a plug-and-play experience when using
standard frameworks such as JDBC \cite{jdbc}. 
Client applications can also use \tool{} as a library in order to directly invoke the storage APIs,
or interact with it via HTTP requests, with JSON payloads.  


MonkeyDB contains a SQL-To-KV compiler that parses an input query\footnote{We
use \url{https://github.com/ballista-compute/sqlparser-rs}}, builds
its Abstract Syntax Tree (AST) and then applies the rewriting steps described in Section~\ref{sec:SQL-to-KV} 
to produce an equivalent sequence of KV API calls ({\tt read()} and {\tt write()}).
It uses a hashing routine ({\tt hash}) to generate unique keys corresponding to each cell in a table.
For instance, in order to insert a value $v$ for a column $c$ in a particular row with primary key value $\pkeyVal$, of a table $\tab$, 
we invoke {\tt write(hash($\tab$, $\pkeyVal$, $c$), $v$)}. 
We currently support only a subset of the standard SQL operators. For instance, 
nested queries or join operators are unsupported; these can be added in the
future with more engineering effort.

MonkeyDB schedules transactions from different sessions
one after the other using a single global lock.
Internally, it maintains execution state as a history consisting of a set of transaction logs, 
write-read relations and a partial session order (as discussed in \sectref{ax-kv}).
On a {\tt read()}, MonkeyDB first collects a set of possible writes present in transaction log 
that can potentially form write-read (read-from) relationships, and then 
invokes the consistency checker (Figure~\ref{fig:block_dia}) to confirm
validity under the chosen isolation level.
Finally, it randomly returns one of the values associated with valid writes.
A user can optionally instruct MonkeyDB to only select
from the set of \textit{latest} valid write per session. This option helps limit weak behaviors
for certain reads.

The implementation of our consistency checker is based on prior work
\cite{DBLP:journals/pacmpl/BiswasE19}. It maintains 
the write-read relation as a graph, and detects cycles (isolation-level
violations) using DFS traversals on the graph. The consistency checker is an independent 
and pluggable module: we have one for Read Committed and one for Causal
Consistency, and more can be added in the future.



\begin{figure}
\includegraphics[scale=0.8]{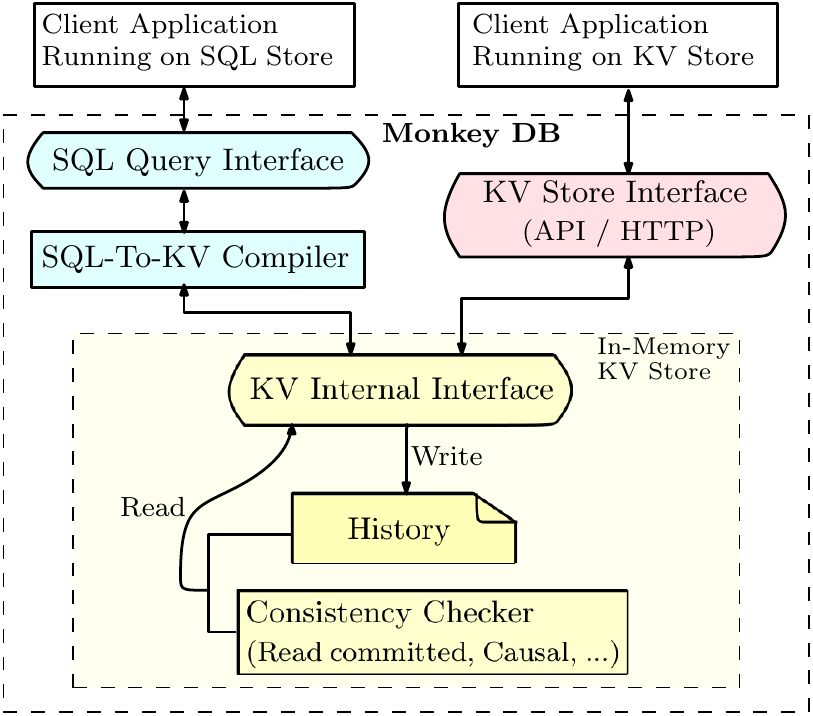}
\caption{Architecture of \tool{}}
\label{fig:block_dia}
\end{figure}



\section{Evaluation: Microbenchmarks}
\label{sec:micro}

We consider a set of micro-benchmarks inspired from real-world applications 
(\sectref{micro-benchmarks}) and evaluate the number of test iterations
required to fail an invalid assertion 
(\sectref{micro-assertion-violations}). We also measure the \textit{coverage} of
weak behaviors provided by MonkeyDB (\sectref{micro-coverage}). Each of these
applications were implemented based on their specifications described in prior
work; they all use MonkeyDB as a library, via its KV interface.  


\subsection{Applications}
\label{sec:micro-benchmarks}

\paragraph{Twitter \cite{twissandra}}
This is based on a social-networking application that allows users to create a new account, follow,
unfollow, tweet, browse the newsfeed (tweets from users you follow)
and the timeline of any particular user. 
\figref{twitter-algo} shows the pseudo code for two operations. 

A user can access twitter from multiple clients (sessions), which could lead to
unexpected behavior under weak isolation levels. 
Consider the following scenario with two users, $A$ and $B$ where user $A$ is accessing twitter from two
different sessions, $S_1$ and $S_2$. User $A$ views the timeline of user $B$
from one session (\texttt{$S_1$:Timeline($B$)}) and decides to follow $B$
through another session (\texttt{$S_2$:Follow($A$, $B$)}). Now when user $A$
visits their timeline or newsfeed (\texttt{$S_2$:NewsFeed($A$)}), they expect to
see all the tweets of $B$ that were visible via \texttt{Timeline} in session $S_1$. But
under weak isolation levels, this does not always hold true and there could
be missing tweets. 

\begin{figure}
  \begin{tabular}{@{\hspace{0ex}}l@{\hspace{-1ex}}l}
	\begin{minipage}{4cm}
		\begin{lstlisting}[basicstyle=\ttfamily\footnotesize,escapeinside={(*}{*)},language=MyLang]
			
// Get user's tweets
Timeline(user u) {
  Begin()
  key = "tweets:" + u.id
  T = read(key) 
  Commit()
  return sortByTime(T)
}
		\end{lstlisting}
	\end{minipage}
  &
	\begin{minipage}{4.3cm}
		\begin{lstlisting}[xleftmargin=3mm,basicstyle=\ttfamily\footnotesize,escapeinside={(*}{*)},language=MyLang]
// Get following users' tweets
NewsFeed(user u) {
  Begin()
  FW = read("following:"+ u.id)
  NF = {}
  foreach v (*$\in$*) FW:
    T = read("tweets:"+ v.id)
    NF = NF (*$\cup$*) T
  Commit()
  return sortByTime(NF)
}
		\end{lstlisting}
	\end{minipage}
\end{tabular}
\vspace{-5ex}
	\caption{Example operations of the Twitter app}
	\label{fig:twitter-algo}
\end{figure}

\vspace{-2mm}
\paragraph{Shopping Cart \cite{sivaramakrishnan2015declarative}}
This application allows a user to add,
remove and change quantity of items from different sessions. It also allows the
user to view all items present in the shopping cart. The pseudo code and
an unexpected behavior under weak isolation levels were discussed in
\sectref{intro}, \figref{motiv}.

\vspace{-2mm}	
\paragraph{Courseware \cite{DBLP:conf/esop/NairP020}}
This is an application for managing students and courses, allowing students
to register, de-register and enroll for courses. Courses can also be created 
or deleted. Courseware maintains the current status of
students (registered, de-registered), courses (active, deleted) as well as
enrollments.
Enrollment can contain only registered students and active courses, subject to the capacity
of the course.

Under weak isolation, it is possible that two different students, when trying to
enroll concurrently, will both succeed even though only one spot was left in the
course. Another example that breaks the application is when a student is trying
to register for a course that is being concurrently removed: once the course is
removed, no student should be seen as enrolled in that course.


\vspace{-2mm}
\paragraph{Treiber Stack \cite{cavNagarMJ20}}
Treiber stack is a concurrent stack data structure that uses 
compare-and-swap (CAS) instructions instead of locks for synchronization. This
algorithm was ported to operate on a kv-store in prior work \cite{cavNagarMJ20}
and we use that implementation. Essentially, the stack contents are placed in
a kv-store, instead of using an in-memory linked data structure.
Each row in the store contains a pair consisting of the stack element and the key of the next
row down in the stack. A designated key ``{\tt head}'' stores the key of the
top of the stack. CAS is implemented as a transaction, but the 
\texttt{pop} and \texttt{push} operations do not use transactions, i.e., each
read/write/CAS is its own transaction.


When two different clients try to \texttt{pop} from the stack concurrently, 
under serializability, each \texttt{pop} would return a unique value, assuming that each pushed value is
unique. However, under causal consistency, concurrent \texttt{pop}s can return the same
value.

\subsection{Assertion Checking}
\label{sec:micro-assertion-violations}

We ran the above applications with MonkeyDB to find out if assertions, capturing
unexpected behavior, were violated under causal consistency. Table
\ref{tab:assert} summarizes the results. For each application, 
we used 3 client threads and 3 operations per thread. 
We ran each test with MonkeyDB for a total of 10,000 times; we refer to a run as
an iteration. We  report the average number of iterations (Iters) 
before an assertion failed, and the corresponding time taken 
(sec). All the assertions were violated within 58 iterations, in half a second
or less. In contrast, running with an actual database almost never
produces an assertion violation.

\begin{table}[]
	\footnotesize
	\begin{tabular}{|l|l|c|c|}
		
		\hline
		
    \textbf{Application} & \textbf{Assertion}   & \multicolumn{2}{c|}{\textbf{Avg. time to fail}}     \\ 
                         &                      & \textbf{(Iters)} & \textbf{(sec)} \\ \hline
		
    Stack                & Element popped more than once  & 3.7  & 0.02 \\ \hline
		
    Courseware           & Course registration overflow & 10.6 & 0.09  \\ \hline
		
    Courseware           & Removed course registration & 57.5 & 0.52   \\ \hline
		
    Shopping            & Item reappears after deletion & 20.2 & 0.14  \\ \hline
		
    Twitter              & Missing tweets in feed & 6.3 & 0.03 \\ \hline
		
	\end{tabular}
	\caption{\label{tab:assert}Assertions checking results in microbenchmarks}
\end{table}

%
%
%
%
%
%
%
%
%
%

\subsection{Coverage}
\label{sec:micro-coverage}

The previous section only checked for a particular set of assertions. As an additional measure of
test robustness, we count the number of distinct \textit{client-observable
states} generated by a test. A client-observable state, for an execution, is the vector of values returned by
read operations. For instance, a stack's state is defined by return values of 
\texttt{pop} operations; a shopping cart's state is defined by the return value
of \texttt{GetCart} and so on. 

For this experiment, we randomly generated test harnesses; each harness spawns
multiple threads that each execute a sequence of operations. In order to compute the absolute maximum
of possible states, we had to limit the size of the tests: either 2 or 3
threads, each choosing between 2 to 4 operations. 

Note that any program that concurrently executes operations against a store has
two main sources of non-determinism: the first is the interleaving of operations
(i.e., the order in which operations are submitted to the store) and second is
the choice of read-from (i.e., the value returned by the store under its
configured isolation level). MonkeyDB only controls the latter; it is up to the
application to control the former. There are many tools that systematically enumerate
interleavings (such as \textsc{Chess} \cite{DBLP:conf/pldi/MusuvathiQ08},
\textsc{Coyote} \cite{coyote-web}), but we use a simple trick
instead to avoid imposing any burden on the application: 
we included an option in MonkeyDB to deliberately add a small random
delay (sleep between $0$ to $4$ ms) before each transaction begins. This option 
was sufficient in our experiments, as we show next.

We also implemented a special setup using the \textsc{Coyote} tool \cite{coyote-web} 
to enumerate all sources of non-determinism, interleavings as well as
read-from, in order to explore the entire state space of a
test. We use this to compute the total number of states. \figref{micro_dfs}
shows the number of distinct 
states observed under different isolation levels, averaged across multiple ($50$) test
harnesses. For each of serializability and causal consistency, we show the max
(as computed by \textsc{Coyote}) and versions with and without the delay option
in MonkeyDB. 

Each of these graphs show similar trends: the number of states with
causal consistency are much higher than with serializability. Thus, testing with a
store that is unable to generate weak behaviors will likely be ineffective.
Furthermore, the ``delay'' versions of MonkeyDB are able to approach the 
maximum within a few thousand attempts, implying that MonkeyDB's strategy of
per-read randomness is effective for providing coverage to the application.

\begin{figure*}[h]
	
	\centering
	
	\includegraphics[width=1.0\textwidth]{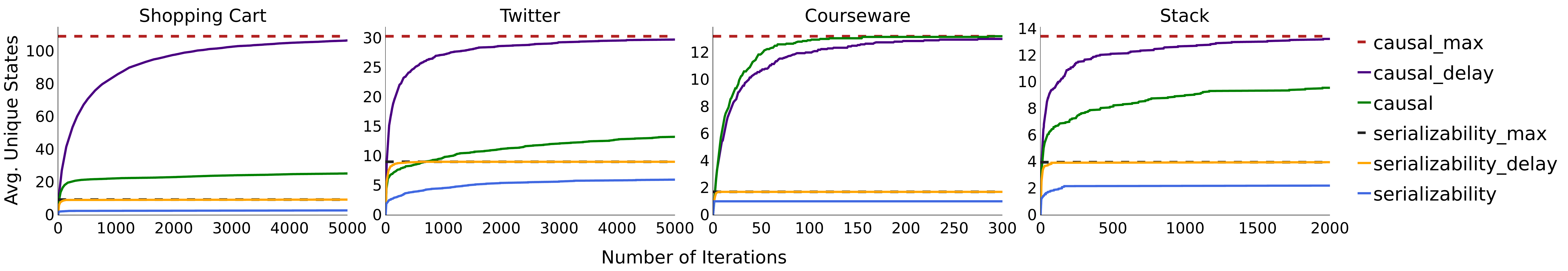}
	
	\caption{State coverage obtained with MonkeyDB for various microbenchmarks}
	
	\label{fig:micro_dfs}
\end{figure*}

\section{Evaluation: OLTP Workloads}
\label{sec:oltp}

OLTPBench \cite{difallah2013oltp} is a benchmark suite of representative 
OLTP workloads for relational databases.
We picked a subset of OLTPBench for which we had reasonable assertions. 
Table~\ref{table-bench} lists basic
information such as the number of database tables, the
number of static transactions, how many of them are read-only, and the number of
different assertions corresponding to system invariants for testing the benchmark. 
We modified OLTPBench by rewriting SQL join and aggregation
operators 
into equivalent application-level
loops, following a similar strategy as prior work \cite{clotho}. Except for
this change, we ran OLTPBench unmodified. 

For TPC-C, we obtained a set of $12$ invariants from its specification
document~\cite{tpcc-spec}. For all other benchmarks, we manually identified 
invariants that the application should satisfy. We asserted these invariants 
by issuing a read-only transaction to \tool{} 
at the end of the execution of the benchmark. 
None of the assertions fail under serializability; they are indeed invariants
under serializability.\footnote{We initially observed two assertions failing
under serializability. Upon analyzing the code, we identified that the
behavior is due to a bug in OLTPBench that we have reported to the authors (link
ommitted).} 
When using weaker isolation, we configured MonkeyDB to use latest reads only
(\sectref{impl}) for the assertion-checking transactions 
in order to isolate the weak behavior to only the application. 

We ran each benchmark $100$ times and report, for each assertion, the number of
runs in which it was violated. Note that OLTPBench runs in two phases. The first
is a loading phase that consists of a big initial transaction to populates tables 
with data, and then the execution phase issues multiple concurrent transactions. 
With the goal of testing correctness, we \textit{turn down} the scale factor to
generate a small load and limit the execution phase time to ten seconds with 
just two or three sessions. A smaller test setup has the advantage
of making debugging easier. With \tool{}, there is no need to generate large
workloads.


\begin{table}
  \footnotesize
	\begin{tabular}{|l|c|c|c|c|}
    \hline
		Benchmark & \#Tables & \#Txns &\#Read-only & \#Assertions \\ \hline
		TPC-C &  9 & 5 & 2& 12\\
		SmallBank & 3 & 6 & 1 & 1\\
		Voter & 3 & 1 & 0 & 1\\
		Wikipedia & 12 & 5 & 2 & 3\\
		\hline
	\end{tabular}	
	\caption{OLTP benchmarks tested with \tool{}}
	\label{table-bench}
\end{table}

\paragraph{TPC-C} TPC-C emulates a wholesale supplier transactional system
that delivers orders for a warehouse company.
This benchmark deals with customers, payments, orders, warehouses, 
deliveries, etc. 
We configured OLTPBench to issue a higher proportion ($>85\%$) of update
transactions, compared to read-only ones.  
Further, we considered a small input workload constituting of one warehouse, two
districts per warehouse and three customers per district.


TPC-C has twelve assertions (A1 to A12) that check for consistency between 
the database tables.
For example, A12 checks: for any customer, the sum of delivered order-line
amounts must be equal to the sum of balance amount and YTD (Year-To-Date)
payment amount of that customer.

    \begin{figure}[t]
        \centering
    	\includegraphics[scale=0.5]{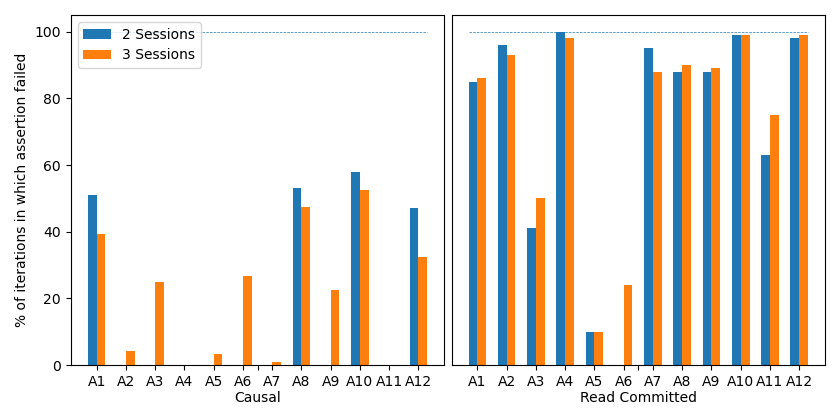}
	    \caption{\small Assertion checking: {TPC-C}}
        \label{fig:tpcc}
    \end{figure}
    \begin{figure}[t]
        \centering
    \includegraphics[scale=0.45]{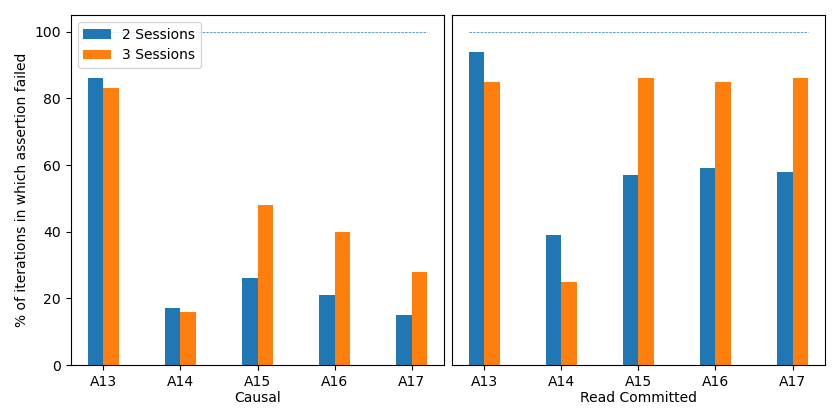}
	    \caption{\small Assertion checking: SmallBank, Voter, and Wikipedia}
    \label{fig:rest}
    \end{figure}

Figure~\ref{fig:tpcc} 
shows the percentage of test runs in which an assertion failed. 
It shows that all the twelve assertions are violated under
Read~Committed isolation level. In fact, $9$ out of the $12$ 
assertions are violated in more than 60\% of the test runs.
In case of causal, 
all assertions are violated
with three sessions, except for A4 and A11. We manually inspected TPC-C and
we believe that both these assertions are valid under causal consistency. 
For instance, A4 checks for consistency between two tables, both of which are
only updated within the same transaction, thus causal consistency is enough to
preserve consistency between them.

These results demonstrate the effectiveness of \tool{} in breaking (invalid) 
assertions. Running with MySQL, under read committed, was unable to violate any assertion except for
two (A10 and A12), even when increasing the number of sessions to $10$. We used
the same time limit of $10$ seconds for the execution phase. We note that MySQL
is much faster than MonkeyDB and ends up processing up to $50\times$ more
transactions in the same time limit, yet is unable to violate most assertions. 
Prior work \cite{clotho} attempted a more sophisticated test
setup where TPC-C was executed on a Cassandra cluster, while running 
Jepsen~\cite{jepsen} for fault injection. This setup also was unable to violate 
all assertions, even when running without transactions, and on a weaker isolation level than read committed. 
Only six assertions were violated with 10 sessions, 
eight assertions with 50 sessions, and ten assertions with 100 sessions.
With \tool{}, there is no need to set up a cluster, use fault injection or
generate large workloads that can make debugging very difficult.

%

\paragraph{SmallBank, Voter, and Wikipedia} 
SmallBank is a standard financial banking system, dealing with customers, saving
and checking accounts, money transfers, etc.  
Voter emulates the voting system of a television show and allows users to vote for their favorite contestants.
Wikipedia is based on the popular online encyclopedia. It deals with a complex database schema involving 
page revisions, page views, user accounts, logging, etc.  
It allows users to edit its pages and maintains a history of page edits and user actions. 

We identified a set of five assertions, A13 to A17, that should be satisfied by these systems.  
For SmallBank, we check if the total money in the bank remains the same while it
is transfered from one account to another (A13). 
Voter requires that the number of votes by a user is limited to a fixed threshold (A14).  
For Wikipedia, we check if for a given user and for a given page, the number of
edits recorded in the user information, history, and logging tables
are consistent (A15-A17). 
As before, we consider small work loads: (1) five customers for SmallBank, (2) one user for Voter, and (3) two pages and two users for Wikipedia.  

Figure~\ref{fig:rest} shows the results. 
\tool{} detected that all the assertions are invalid under the chosen isolation levels.
Under causal, \tool{} could break an assertion in 26.7\% (geo-mean) runs given 2 sessions and in 37.2\% (geo-mean) runs given 3 sessions.
Under read committed, the corresponding numbers are 56.1\% and 65.4\% for 2 and
3 sessions, respectively.


\section{Related Work}
\label{sec:related}


There have been several directions of work addressing the correctness of database-backed applications. 
We directly build upon one line of work concerned with the logical formalization
of isolation levels 
\cite{ansi,DBLP:conf/icde/AdyaLO00,DBLP:conf/sigmod/BerensonBGMOO95,DBLP:conf/concur/Cerone0G15,DBLP:journals/pacmpl/BiswasE19}.
Our work relies on the axiomatic definitions of isolation levels, as given
in~\cite{DBLP:journals/pacmpl/BiswasE19}, which also investigated
the problem of checking whether a given history satisfies a certain isolation
level. Our kv-store implementation relies on these algorithms 
to check the validity of the values returned by read operations. Working with a
logical formalization allowed us to avoid implementing an actual database with replication or
sophisticated synchronization.

Another line of work concentrates on the problem of finding ``anomalies'': 
behaviors that are not possible under serializability. This is typically done
via a static analysis of the application code that builds a static dependency graph that
over-approximates the data dependencies in all possible
executions of the application~\cite{DBLP:journals/jacm/CeroneG18,DBLP:journals/jacm/CeroneG18,DBLP:conf/concur/0002G16,DBLP:journals/tods/FeketeLOOS05,DBLP:conf/vldb/JorwekarFRS07,acidrain,isodiff}.
Anomalies with respect to a given isolation level then corresponds to a
particular class of cycles in this graph. Static dependency graphs turn out to
be highly imprecise in representing feasible executions, leading to false
positives. Another source of false positives is that an anomaly might not be a
bug because the application may already be designed to handle the
non-serializable behavior \cite{DBLP:conf/pldi/BrutschyD0V18,isodiff}. 
Recent work has tried to address these issues by using more precise 
logical encodings of the application,
e.g.~\cite{DBLP:conf/popl/BrutschyD0V17,DBLP:conf/pldi/BrutschyD0V18} or
by using user-guided heuristics~\cite{isodiff}. 

Another approach consists of modeling the application
logic and the isolation level in first-order logic and relying on SMT solvers to
search for anomalies~\cite{DBLP:journals/pacmpl/KakiESJ18,DBLP:conf/concur/NagarJ18,burcu-netys},
or defining specialized reductions to assertion
checking~\cite{DBLP:conf/concur/BeillahiBE19,DBLP:conf/cav/BeillahiBE19}.
The \textsc{Clotho} tool \cite{clotho}, for instance, uses a static analysis of the application to
generate test cases with plausible anomalies, which are deployed in a concrete
testing environment for generating actual executions. 

Our approach, based on testing with MonkeyDB, has several practical advantages.
There is no need for analyzing application code; we can work with any
application. There are no false positives because we directly run the
application and check for user-defined assertions, instead of looking for
application-agnostic anomalies. The limitation, of course, is
the inherent incompleteness of testing.

Several works have looked at the problem of reasoning about the correctness of
applications executing under weak isolation and introducing additional
synchronization when
necessary~\cite{DBLP:conf/eurosys/BalegasDFRPNS15,DBLP:conf/popl/GotsmanYFNS16,DBLP:conf/esop/NairP020,DBLP:conf/usenix/0001LCPRV14}.
As in the previous case, our work based on testing has the advantage that it can
scale to real sized applications (as opposed to these techniques which are based
on static analysis or logical proof arguments), but it cannot prove that an
application is correct. Moreover, the issue of repairing applications is
orthogonal to our work. 

From a technical perspective, our operational semantics based on recording past
operations and certain data-flow and control-flow dependencies is similar to
recent work on stateless model checking in the context of weak memory
models,
e.g.~\cite{DBLP:journals/pacmpl/Kokologiannakis18,DBLP:conf/tacas/AbdullaAAJLS15}.
This work, however, does not consider transactions. Furthermore, their focus is on
avoiding enumerating equivalent executions, which is beyond the scope of our
work (but an interesting direction for future work).

\section{Conclusion}
\label{sec:conc}

Our goal is to enable developers to test the correctness of their storage-backed applications under 
weak isolation levels. Such bugs are hard to catch because weak behaviors are
rarely generated by real storage systems, but failure to address them can lead
to loss of business \cite{acidrain}. We present MonkeyDB, an easy-to-use mock storage system
for weeding out such bugs. MonkeyDB uses a logical understanding of isolation
levels to provide (randomized) coverage of all possible weak behaviors. Our evaluation reveals that
using MonkeyDB is very effective at breaking assertions that would otherwise
hold under a strong isolation level.


\newpage

\bibliography{main,dblp}

\newpage
\appendix
\section{Correctness of the Operational Semantics}

\begin{figure} [t]
\small
  \centering
  \begin{mathpar}
    \inferrule[spawn*]{\tr \mbox{ fresh}\quad \mathsf{P}(j) = \ibegin; \mathsf{Body}; \icommit; \mathsf{S} \quad \vec{\mathsf{B}}(j) = \epsilon}{
      \hist,\vec{\gamma},\vec{\mathsf{B}},\mathsf{P}
      \Rightarrow
      \hist \oplus_j \tup{\tr,\emptyset,\emptyset},\vec{\gamma}[j\mapsto \emptyset],\vec{\mathsf{B}}[j\mapsto \mathsf{Body}],\mathsf{P}[j\mapsto \mathsf{S}]
    } 

    \inferrule[if-true]{\varphi(\vec{x})[x\mapsto \vec{\gamma}(j)(x): x\in\vec{x}]\mbox{ true} \\
    \vec{\mathsf{B}}(j) = \iif{\phi(\vec{x})}{\mathsf{Instr}};\mathsf{B}
    }{
      \hist,\vec{\gamma},\vec{\mathsf{B}}, \mathsf{P}
      \Rightarrow
      \hist,\vec{\gamma},\vec{\mathsf{B}}[j\mapsto \mathsf{Instr};\mathsf{B}],\mathsf{P}
    } 

    \inferrule[if-false]{\varphi(\vec{x})[x\mapsto \vec{\gamma}(j)(x): x\in\vec{x}]\mbox{ false} \\
    \vec{\mathsf{B}}(j) = \iif{\phi(\vec{x})}{\mathsf{Instr}};\mathsf{B}
    }{
      \hist,\vec{\gamma},\vec{\mathsf{B}}, \mathsf{P}
      \Rightarrow
      \hist,\vec{\gamma},\vec{\mathsf{B}}[j\mapsto \mathsf{B}],\mathsf{P}
    } 

    \inferrule[write]{v = \vec{\gamma}(j)(x)\quad \id\mbox{ fresh} \quad 
    \vec{\mathsf{B}}(j) = \iwrite(\key,\xvar);\mathsf{B}
    }{
      \hist,\vec{\gamma},\vec{\mathsf{B}}, \mathsf{P}
      \Rightarrow
      \hist \oplus_j \wrt[\id]{\key}{\val},\vec{\gamma},\vec{\mathsf{B}}[j\mapsto \mathsf{B}], \mathsf{P}
    } 

    \inferrule[read-local]{
    \wrt{\key}{\val}\mbox{ is the last write on $\key$ in $\tr$}\\
    \id\mbox{ fresh } \\
    \vec{\mathsf{B}}(j) = \xvar := \iread(\key);\mathsf{B}
    }{
      \hist,\vec{\gamma},\vec{\mathsf{B}}, \mathsf{P}
      \Rightarrow
      \hist \oplus_j \rd[\id]{\key}{\val},\vec{\gamma}[(j,\xvar)\mapsto \val],\vec{\mathsf{B}}[j\mapsto \mathsf{B}],\mathsf{P}
    } 

    \inferrule[read-extern*]{
    \vec{\mathsf{B}}(j) = \xvar := \iread(\key);\mathsf{B} \\
    \hist=(T,\so,\wro) \\
    \tr \mbox{ is the id of the last transaction log in $\so(j)$} \\
    \wrt{\key}{\val}\in\writeOp{\tr'}\mbox{ with $\tr'\in \transC{\hist,\vec{\mathsf{B}}}$ and $\tr\neq \tr'$} \\
    \id\mbox{ fresh }\\
    \hist' = (\hist \oplus_j \rd[\id]{\key}{\val}) \oplus \wro(\tr',\rd[i]{\key}{\val}) }{
      \hist,\vec{\gamma},\vec{\mathsf{B}}, \mathsf{P}
      \Rightarrow
      \hist',\vec{\gamma}[(j,\xvar)\mapsto \val],\vec{\mathsf{B}}[j\mapsto \mathsf{B}],\mathsf{P}
    } 
    
  \end{mathpar}
  \caption{A baseline operational semantics for $\KVProgs$ programs. Above, $\transC{\hist,\vec{\mathsf{B}}}$ denotes the set of transaction logs in $\hist$ that excludes those corresponding to live transactions, i.e., transaction logs $\tr''\in T$ such that $\tr''$ is the last transaction log in some $\so(j')$ and $\vec{B}(j')\neq\epsilon$.}
  \label{fig:op:sem:baseline:complete}
\end{figure}

This section provides more details about the proof of correctness for our operational semantics defined in Figure~\ref{fig:op:sem}. The complete definition of the baseline semantics is given in Figure~\ref{fig:op:sem:baseline:complete}.

The notion of prefix of a tuple $\{\hist_2,\co_2\}$ is formally defined as follows. For a relation $R\subseteq A\times B$, the restriction of $R$ to $A'\times B'$, denoted by $R\downarrow A'\times B'$, is defined by $\{(a,b): (a,b)\in R, a\in A', b\in B'\}$.
For $\hist_1=\tup{T_1, \so_1, \wro_1}$ and $\hist_2=\tup{T_2, \so_2, \wro_2}$, $\tup{\hist_1,\co_1}$ is a \emph{prefix} of $\tup{\hist_2,\co_2}$, denoted by $\tup{\hist_1,\co_1}\leq \tup{\hist_2,\co_2}$, iff $\tup{\hist_1,\co_1}\leq \tup{\hist_2,\co_2}$ iff $T_1=T_1'\cup\{\tup{t,O,\po}\}$, $T_2=T_2'\cup\{\tup{t,O',\po'}\}$, $T_1'\subseteq T_2'$, $O\subseteq O'$, $\po = \po' \downarrow O\times O$, $\so_1 = \so_2 \downarrow T_1\times T_1$, $\wro_1= \wro_2\downarrow T_1\times \readOp{T_1}$, and $\co_1= \co_2\downarrow T_1\times T_1$.

Then, a property $\phi(\tr_2,\alpha)$ used to define an axiom like in (\ref{eq:axiom}), is called \emph{monotonic} iff 
for every  $\tup{\hist_1,\co_1}\leq \tup{\hist_2,\co_2}$, 
\begin{align*}
\forall \tr_2, \forall \alpha. \tup{\hist_1,\co_1}\models \phi(\tr_2,\alpha) \Rightarrow  \tup{\hist_2,\co_2}\models \phi(\tr_2,\alpha).
\end{align*}

\begin{lemma}\label{lem:prefix}
For any monotonic axiom $X$, if $\tup{\hist_1,\co_1}\leq \tup{\hist_2,\co_2}$, then
\begin{align*}
\tup{\hist_2,\co_2}\mbox{ satisfies } X \Rightarrow \tup{\hist_1,\co_1}\mbox{ satisfies } X
\end{align*}
\end{lemma}
\begin{proof}(Sketch)
Given a monotonic axiom, the number of instantiations of $\forall k$, $\forall t_1$, $\forall t_2$, and $\forall \alpha$ from (\ref{eq:axiom}) that satisfy the left-hand side of the entailement in the context of $\tup{\hist_1,\co_1}$ is a subset of the same type of instantiations in the context of $\tup{\hist_2,\co_2}$. Therefore, the $\co$ constraints imposed in the context of $\tup{\hist_1,\co_1}$ (by the right-hand side of the entailement) are a subset of the $\co$ constraints imposed in the context of $\tup{\hist_2,\co_2}$. Since the latter are satisfied (because $\tup{\hist_2,\co_2}$  satisfies  $X$), the former are also satisfied and hence, $\tup{\hist_1,\co_1}$  satisfies  $X$.
\end{proof}

Lemma~\ref{lem:prefix} extends obviously to isolation levels defined as conjunctions of axioms (which is the case for all the isolation levels that we are aware of~\cite{DBLP:journals/pacmpl/BiswasE19}).

 \begin{theorem}
For any isolation level $I$ defined by a set of monotonic axioms,
\begin{align*}
\histOf[I]{\prog} = \{ h \in \histOf[*]{\prog}: h\mbox{ satisfies }I\}.
\end{align*}
 \end{theorem}
 \begin{proof}(Sketch)
For the direction $\subseteq$, let $c_0 c_1\ldots c_n$ be an execution under $\Rightarrow_I$, where $c_n$ is a final configuration. We need to show that the history $\hist_n$ contained in $c_n$ belongs to $\histOf[*]{\prog}$ and that it satisfies $I$. The fact that $\hist_n\in \histOf[*]{\prog}$ is a direct consequence of the fact that $\Rightarrow_I$ is more constrained than $\Rightarrow$. To prove that $\hist_n$ satisfies $I$, let $c_j$ be the latest configuration in the execution that is obtained from $c_{j-1}$ through an application of \textsc{read-extern}. By the definition of this rule, the history $\hist_j$ in $c_j$ satisfies $I$. Since the write-read relation of $\hist_j$ is identical to that of $\hist_n$, any axiom of the form (\ref{eq:axiom}) satisfied by $\hist_j$ is also satisfied by $\hist_n$ (the set of instantiations of $\forall \tr_1$ and $\forall \alpha$ in (\ref{eq:axiom}) that satisfy the left part of the entailment are the same in $\hist_j$ and $\hist_n$). Therefore, $\hist_n$ satisfies $I$, which concludes this part of the proof.

For the reverse, let $\hist=\tup{T, \so, \wro}\in \histOf[*]{\prog}$ that satisfies $I$. Since $\hist$ satisfies $I$, there exists a commit order $\co$ such that $\wro\cup\so\subseteq \co$ and $\tup{h,\co}$ satisfies the axioms defining $I$. We show that there exists an execution $c_0 c_1\ldots c_n$ under $\Rightarrow_I$ where transactions are executed serially in the order defined by $\co$, such that $c_n$ is a final configuration that contains $\hist$. The only difficulty is showing that the \textsc{read-extern} transitions between two configurations $c_j$ and $c_{j+1}$ that add a write-read dependency $(\tr',\rd{\key}{\val})\in\wro$ are enabled even though the transaction log $\tr$ containing $\rd{\key}{\val}$ is ``incomplete'' in the history $\hist_j$ of $c_j$, and $\hist_j$ does not contain transactions committed after $\tr$. This relies on the prefix-closure property in Lemma~\ref{lem:prefix}.
Let $\co_j$ be the order in which transactions have been executed until $c_j$. Then, $\tup{\hist_j,\co_j}$ is a prefix of $\tup{\hist,\co}$, and $\tup{\hist_j,\co_j}\models I$ because $\tup{\hist,\co}\models I$.
 \end{proof}

\end{document}